\documentclass{article}

\pdfoutput=1 

\usepackage{geometry}
\usepackage[pdftex]{graphicx}
\DeclareGraphicsExtensions{.pdf}

\usepackage{amssymb}
\usepackage{amsmath}
\usepackage{bbold} %

\usepackage{wasysym} %

\usepackage{fancyhdr} %
\usepackage{ifthen} %
\usepackage{url}
\usepackage[latin1]{inputenc} %

\def\ba#1\ea{\begin{align*}#1\end{align*}} %
\def\banum#1\eanum{\begin{align}#1\end{align}} %

\makeatletter
\@ifclassloaded{llncs}{
}{
\newtheorem{lemma}{Lemma} 
\newtheorem{proposition}[lemma]{Proposition} 

\newtheorem{theorem}[lemma]{Theorem}

\newenvironment{proof}{\par\noindent{\em Proof. }}{\hfill$\Box$\\[2mm]}
} 
\newcommand{\RR}{\mathbb{R}}
 
\newcommand{\Nat}{\mathbb{N}} 
\newcommand{\R}{\RR}

\newcommand{\Lcal}{\mathcal{L}}

\newcommand{\Vcal}{\mathcal{V}}

\newcommand{\Xcal}{\mathcal{X}}

\DeclareMathOperator{\vol}{vol} %

\DeclareMathOperator{\Tr}{Tr}%

\newcommand{\norm}[1][\cdot]{\|#1\|}

\newcommand{\eps}{\ensuremath{\varepsilon}}
\renewcommand{\epsilon}{\ensuremath{\varepsilon}}
\newcommand{\charfct}{\mathbb{1}}

\newcommand{\union}{\;\cup\;}

\newcommand{\blobb}[1]{%
\begin{list}{$\bullet$}{%
\setlength{\topsep}{0cm}
\setlength{\partopsep}{-\parskip}
\setlength{\leftmargin}{0.5cm}
}{\item #1}%
\end{list}
}

\usepackage{apacite}
\renewcommand{\BBAA}{and}

\newcommand{\Lun}{L}
\newcommand{\Lsym}{L_{\text{sym}}}
\newcommand{\Lrw}{L_{\text{rw}}}

\renewcommand{\bar}{\overline} %

\DeclareMathOperator{\Ratiocut}{RatioCut}
\DeclareMathOperator{\Ncut}{Ncut}
\DeclareMathOperator{\Cut}{cut}

\DeclareMathOperator{\MinmaxCut}{MinMaxCut}

\newcommand{\ulesalgorithm}[2]{%
\vspace{1cm}
\fbox{\parbox{\textwidth}{%
{\bf #1}\\[2mm]
{\tt #2 }
}}
\vspace{1cm}
}

\usepackage{color}
\begin{document}

\title{A Tutorial on Spectral Clustering}
\author{Ulrike von Luxburg\\
Max Planck Institute for Biological Cybernetics\\
Spemannstr. 38, 72076 T{\"u}bingen, Germany\\
\url{ulrike.luxburg@tuebingen.mpg.de}
}
\date{\small This article appears in Statistics and Computing, 17
  (4), 2007. \\The original publication is available at \url{www.springer.com}.}
\maketitle

\begin{abstract} 
  In recent years, spectral clustering has become one of the most
  popular modern clustering algorithms. It is simple to implement, can
  be solved efficiently by standard linear algebra software, and very
  often outperforms traditional clustering algorithms such as the
  k-means algorithm. On the first glance spectral clustering appears
  slightly mysterious, and it is not obvious to see why it works at
  all and what it really does. The goal of this tutorial is to give
  some intuition on those questions.  We describe different graph
  Laplacians and their basic properties, present the most common
  spectral clustering algorithms, and derive those algorithms from
  scratch by several different approaches.  Advantages and
  disadvantages of the different spectral clustering algorithms are
  discussed.
\end{abstract}


\section{Introduction} \label{sec-intro}

Clustering is one of the most widely used techniques for exploratory data
analysis, with applications ranging from statistics, computer science,
biology to social sciences or psychology. In virtually
every scientific field dealing with empirical data, people attempt to get
a first impression on their data by trying to identify groups of
``similar behavior'' in their data. 
In this article we would like to introduce the reader to the family of
spectral clustering algorithms. Compared to the ``traditional
algorithms'' such as $k$-means or single linkage, spectral clustering
has many fundamental advantages. Results obtained by spectral
clustering often outperform the traditional approaches, spectral
clustering is very simple to
implement and can be solved efficiently by standard linear algebra methods. \\

This tutorial is set up as a self-contained introduction to spectral
clustering. We derive spectral clustering from scratch and present
different points of view to why spectral clustering works.
Apart from basic linear algebra, no particular mathematical background
is required by the reader. However, we do not attempt to give a
concise review of the whole literature on spectral clustering, which
is impossible due to the overwhelming amount of
literature on this subject.  The first two sections are devoted to a
step-by-step introduction to the mathematical objects used by spectral
clustering: similarity graphs in Section \ref{sec-graphs}, and graph
Laplacians in Section \ref{sec-laplacians}. The spectral clustering
algorithms themselves will be presented in Section
\ref{sec-algorithms}.  The next three sections are then devoted to
explaining why those algorithms work. Each section corresponds to one
explanation: Section \ref{sec-graphcuts} describes a graph
partitioning approach, Section \ref{sec-rw} a random walk perspective,
and Section \ref{sec-perturbation} a perturbation theory approach. In
Section \ref{sec-inpractice} we will study some practical issues
related to spectral clustering, and discuss various extensions and
literature related to spectral clustering in Section
\ref{sec-extensions}.

\section{Similarity graphs} \label{sec-graphs}

Given a set of data points $x_1,\hdots x_n$ and some notion of
similarity $s_{ij} \geq 0$ between all pairs of data points $x_i$ and
$x_j$, the intuitive goal of clustering is to divide the data points
into several groups such that points in the same group are similar and
points in different groups are dissimilar to each other. If we do not
have more information than similarities between data points, a nice
way of representing the data is in form of the {\em similarity graph}
$G = (V,E)$. Each vertex $v_i$ in this graph represents a data point
$x_i$.  Two vertices are connected if the similarity $s_{ij}$ between
the corresponding data points $x_i$ and $x_j$ is positive or larger
than a certain threshold, and the edge is weighted by $s_{ij}$.  The
problem of clustering can now be reformulated using the similarity
graph: we want to find a partition of the graph such that the edges
between different groups have very low weights (which means that
points in different clusters are dissimilar from each other) and the
edges within a group have high weights (which means that points within
the same cluster are similar to each other). To be able to formalize
this intuition we first want to introduce some basic graph notation
and briefly discuss the kind of graphs we are going to study.

\subsection{Graph notation}

Let $G = (V,E)$ be an undirected graph with vertex
set $V = \{v_1, \hdots, v_n\}$.
In the following we assume that the graph $G$ is weighted, that is
each edge between two vertices $v_i$ and $v_j$ carries a non-negative weight $w_{ij} \geq
0$. 
The weighted {\em adjacency matrix} of the graph is the matrix $W = (w_{ij})_{i,j =
  1,\hdots,n}$. If $w_{ij} = 0$ this means that the vertices $v_i$ and
$v_j$ are not connected by an edge. As $G$ is undirected we require $w_{ij} = w_{ji}$. 
The degree of a vertex $v_i \in V$ is defined as 
\ba
d_i = \sum_{j = 1}^n w_{ij}.
\ea
Note that, in fact, this sum only runs over all vertices adjacent to $v_i$, as
for all other vertices $v_j$ the weight $w_{ij}$ is 0. The {\em degree matrix} $D$
is defined as the diagonal matrix with the degrees $d_1,\hdots,d_n$ on the diagonal.
Given a subset of vertices $A \subset V$, we denote its complement $V
\setminus A$ by $\bar A$.  We define the indicator vector $\charfct_A
= (f_1, \hdots, f_n)' \in \R^n$ as the vector with entries $f_i = 1$
if $v_i \in A$ and $f_i = 0$ otherwise. For convenience we introduce
the shorthand notation ${i\in A}$ for the set of indices $\{i \;| \;
v_i \in A\}$, in particular when dealing with a sum like $\sum_{i \in
  A} w_{ij}$.  For two not necessarily disjoint sets $A,B \subset V$ we
define 
\ba
W(A,B):= \sum_{i \in A, j \in B} w_{ij}.
\ea
 We consider two different ways of measuring the ``size'' of a
subset $A \subset V$:
\ba 
& |A| := \text{ the number of vertices in } A \\
& \vol(A) := \sum_{i \in A} d_i. 
\ea 
Intuitively, $|A|$  measures the size of $A$ by its number of
vertices, while $\vol(A)$ measures the size of $A$ by summing over the weights
of all edges attached to vertices in $A$. 
A subset $A \subset V$ of a graph is connected if any two vertices in
$A$ can be joined by a path such that all intermediate points also lie
in $A$. A subset $A$ is called a connected component if it is
connected and if there are no connections between vertices in $A$ and
$\bar A$. The nonempty sets $A_1, \hdots, A_k$ form a partition of the
graph if $A_i \cap A_j = \emptyset$ and $A_1 \union \hdots \union A_k
= V$.

  \subsection{Different similarity graphs} \label{ssec-sim-graphs}

  There are several popular constructions to transform a given set
  $x_1,\hdots,x_n$ of data points with pairwise similarities $s_{ij}$
  or pairwise distances $d_{ij}$ into a graph. When constructing
  similarity graphs the goal is to model the local neighborhood
  relationships between the data points. \\

  {\bf The $\eps$-neighborhood graph:} Here we connect all points
  whose pairwise distances are smaller than $\eps$.  As the distances
  between all connected points are roughly of the same scale (at most
  $\eps$), weighting the edges would not incorporate more information
  about the data to the graph. Hence, the $\eps$-neighborhood graph is
  usually considered as an unweighted graph. \\

{\bf $k$-nearest neighbor graphs:} Here the goal is to connect vertex
$v_i$ with vertex $v_j$ if $v_j$ is among the $k$-nearest neighbors of
$v_i$. However, this definition leads to a directed graph, as the
neighborhood relationship is not symmetric. There are two ways of
making this graph undirected. The first way is to simply ignore the
directions of the edges, that is we connect $v_i$ and $v_j$ with an
undirected 
edge if $v_i$ is among the $k$-nearest neighbors of $v_j$ {\em or} if
$v_j$ is among the $k$-nearest neighbors of $v_i$. The resulting graph
is what is usually called {\em the $k$-nearest neighbor graph}. The
second choice is to connect vertices $v_i$ and $v_j$ if both $v_i$ is
among the $k$-nearest neighbors of $v_j$ {\em and} $v_j$ is among the
$k$-nearest neighbors of $v_i$. The resulting graph is called the {\em
  mutual $k$-nearest neighbor graph}.  In both cases, after connecting
the appropriate vertices we weight the edges by the similarity of their endpoints. \\

{\bf The fully connected graph:} Here we simply connect all points
with positive similarity with each other, and we weight all edges by
$s_{ij}$. As the graph should represent the local neighborhood
relationships, this construction is only useful if the
similarity function itself models local neighborhoods.  An example for
such a similarity function is the Gaussian similarity function
$s(x_i,x_j) = \exp(- {\norm[x_i - x_j]^2}/({2 \sigma^2}))$, where
the parameter $\sigma$ controls the width of the neighborhoods. This
parameter plays a similar role as the parameter $\eps$ in case
of the $\eps$-neighborhood graph. \\

All graphs mentioned above are regularly used in spectral
clustering. To our knowledge, theoretical results on the question how
the choice of the similarity graph influences the spectral clustering
result do not exist. For a discussion of the behavior of
the different graphs we refer to Section \ref{sec-inpractice}.\\

\section{Graph Laplacians and their basic properties} \label{sec-laplacians}

The main tools for spectral clustering are graph Laplacian
matrices.  There exists a whole field dedicated to the study of those
matrices, called spectral graph theory (e.g., see \citeNP{Chung97}). In
this section we want to define different graph Laplacians and point
out their most important properties. We will carefully distinguish
between different variants of graph Laplacians. Note that in the
literature there is no  unique convention which matrix exactly is
called ``graph Laplacian''. 
Usually, every author just calls ``his'' matrix the graph
Laplacian. Hence, a lot of care is needed when reading literature on
graph Laplacians.\\

In the following we always assume that $G$ is an undirected, weighted
graph with weight matrix $W$, where $w_{ij} = w_{ji} \geq 0$. When
using eigenvectors  of a matrix, we will not necessarily assume that
they are normalized. For example, the constant vector
$\charfct$ and a multiple $a \charfct$ for some $a \neq 0$ will be
considered as the same eigenvectors. Eigenvalues will always be
ordered increasingly, respecting multiplicities. By ``the first $k$
eigenvectors'' we refer to the eigenvectors corresponding to the $k$
smallest eigenvalues. 

\subsection{The unnormalized graph Laplacian} \label{ssec-lun}

The unnormalized graph Laplacian matrix is defined as
\ba
\Lun = D - W.
\ea
An overview over many of its
properties can be found in \citeA{Mohar91,Mohar97}.  The following
proposition summarizes the most important facts needed for spectral
clustering.

\begin{proposition}[Properties of $\Lun$] 
\label{prop-unnorm}
The matrix $\Lun$ satisfies the following properties: 
\begin{enumerate}
\item For every vector $f \in \R^n$ we have 
\ba
f'\Lun f = \frac{1}{2} \sum_{i,j=1}^n w_{ij} (f_i - f_j)^2.
\ea

\item $\Lun$ is symmetric and positive semi-definite.

\item The smallest eigenvalue of $\Lun$ is 0, the corresponding
  eigenvector is the constant one vector $\charfct$. 

\item $\Lun$ has $n$ non-negative, real-valued eigenvalues $0 =
  \lambda_1 \leq \lambda_2 \leq \hdots \leq \lambda_n$.

\end{enumerate}
\end{proposition}
\begin{proof}\\
Part (1): By the definition of $d_i$, 
\ba
& f' \Lun f = f'D f - f'W f = 
\sum_{i=1}^n d_i f_i^2 - \sum_{i,j=1}^n f_i f_j w_{ij} \\
& = \frac{1}{2}\left( \sum_{i=1}^n d_i f_i^2 - 2 \sum_{i,j=1}^n f_i f_j w_{ij} +
  \sum_{j=1}^n d_j f_j^2 \right)
=  \frac{1}{2} \sum_{i,j=1}^n w_{ij} (f_i - f_j)^2.
\ea
Part (2): The symmetry of $\Lun$ follows directly from the symmetry of
$W$ and $D$. The positive semi-definiteness is a direct consequence of
Part (1), which shows that $f'\Lun f \geq 0$ for all $f \in \R^n$. \\
Part (3): Obvious.\\
Part (4) is a direct consequence of Parts (1) - (3). 
\end{proof}

Note that the unnormalized graph Laplacian does not depend on the
diagonal elements of the adjacency matrix $W$. Each adjacency matrix which
coincides with $W$ on all off-diagonal positions leads to the same
unnormalized graph Laplacian $L$. In particular, self-edges in a
graph do not change the corresponding graph Laplacian. \\

The unnormalized graph Laplacian and its eigenvalues and eigenvectors
can be used to describe many properties of graphs, see
\citeA{Mohar91,Mohar97}. One example which will be important for spectral
clustering is the following:

\begin{proposition}[Number of connected components and the spectrum of $\Lun$]
\label{prop-connected-unnorm}
Let $G$ be an undirected graph with non-negative weights. Then the
multiplicity $k$ 
of the eigenvalue $0$ of $\Lun$ equals the number
of connected components $A_1, \hdots, A_k$ in the graph. The eigenspace of eigenvalue $0$ is spanned
by the indicator vectors $ \charfct_{A_1}, \hdots, \charfct_{A_k}$ of those components. 
\end{proposition}
\begin{proof} 
We start with the
  case $k=1$, that is the graph is connected.  Assume that $f$ is an
  eigenvector with eigenvalue $0$. Then we know that
\ba
0 = f' \Lun f = \sum_{i,j=1}^n w_{ij} (f_i - f_j)^2.
\ea
As the weights $w_{ij}$ are non-negative, this sum can only vanish if
all terms $w_{ij}(f_i - f_j)^2$ vanish. Thus, if two vertices $v_i$
and $v_j$ are connected (i.e., $w_{ij} > 0$), then $f_i$ needs to
equal $f_j$. With this argument we can see that $f$ needs to be
constant for all vertices which can be connected by a path in the
graph. Moreover, as all vertices of a connected component in an
undirected graph can be connected by a path, $f$ needs to be constant
on the whole connected component. In a graph consisting of only one
connected component we thus only have the constant one vector
$\charfct$ as eigenvector with eigenvalue 0, which obviously is the
indicator vector of the connected component.
\\

Now consider the case of $k$ connected components. Without loss of
generality we assume that the vertices are ordered according to the
connected components they belong to. In this case, the adjacency
matrix $W$ has a block diagonal form, and the same is true for the
matrix $\Lun$: 
\ba
\Lun = 
\begin{pmatrix}
L_1 &     &        &    \\
    & L_2 &        &    \\
    &     & \ddots &    \\
    &     &        & L_k\\
\end{pmatrix}
\ea
Note that each of the blocks $L_i$ is a proper graph Laplacian on its
own, namely the Laplacian corresponding to the subgraph of the $i$-th connected
component. 
As it is the case for all block diagonal matrices, we know that the
spectrum of $\Lun$ is given by the union of the spectra of $L_i$, and the
corresponding eigenvectors of $\Lun$ are the eigenvectors of $L_i$,
filled with 0 at the positions of the other blocks. 
As each $L_i$ is a graph Laplacian of a connected graph, we know that
every $L_i$ has eigenvalue 0 with multiplicity 1, and the
corresponding eigenvector is the constant one vector on the $i$-th
connected component. Thus, the matrix $\Lun$ has as many eigenvalues
$0$ as there are connected components, and the corresponding
eigenvectors are the indicator vectors of the connected components.
\end{proof}
\subsection{The normalized graph Laplacians} \label{ssec-lnorm}

There are two matrices which are called normalized graph
Laplacians in the literature. Both matrices are closely related to
each other and are defined as 
\ba
& \Lsym := D^{-1/2} \Lun D^{-1/2} = I -  D^{-1/2} W D^{-1/2}\\
& \Lrw := D^{-1} \Lun = I - D^{-1} W. 
\ea 
We denote the first matrix by $\Lsym$ as it is a symmetric matrix, and
the second one by $\Lrw$ as it is closely related to a random walk.
In the following we summarize several properties of $\Lsym$ and
$\Lrw$. The standard reference for normalized graph Laplacians is \citeA{Chung97}. \\

\begin{proposition}[Properties of $\Lsym$ and $\Lrw$] 
\label{prop-norm}
The normalized Laplacians satisfy the following properties:
\begin{enumerate}

\item For every $f \in \R^n$ we have 
\ba
f' \Lsym f = 
\frac{1}{2} \sum_{i,j=1}^n w_{ij} 
\left(\frac{f_i}{\sqrt{d_i}} - \frac{f_j}{\sqrt{d_j}}\right)^2.
\ea

\item $\lambda$ is an eigenvalue of $\Lrw$ with eigenvector $u$ if
  and only if $\lambda$ is an eigenvalue of $\Lsym$ with eigenvector
  $w = D^{1/2} u$. 

\item $\lambda$ is an eigenvalue of $\Lrw$ with eigenvector $u$ if and
  only if $\lambda$ and $u$ solve the generalized eigenproblem
  $\Lun u = \lambda D u$. 

\item $0$ is an eigenvalue of $\Lrw$ with the constant one vector
  $\charfct$ as eigenvector. $0$ is an eigenvalue of $\Lsym$ with
  eigenvector $D^{1/2}\charfct$. 

\item $\Lsym$ and $\Lrw$ are positive semi-definite and have $n$
  non-negative   real-valued eigenvalues  $0 = \lambda_1 \leq \hdots
  \leq \lambda_n$. 

\end{enumerate}
\end{proposition}
\begin{proof}
Part (1) can be proved similarly to Part (1) of Proposition
\ref{prop-unnorm}. \\
Part (2) can be seen immediately by multiplying the eigenvalue
equation $\Lsym w = \lambda w$ with $D^{-1/2}$ from the left and
substituting $u = D^{-1/2}w$. \\
Part (3) follows directly by multiplying the eigenvalue equation $\Lrw
u = \lambda u$ with $D$ from the left. \\
Part (4): The first statement is obvious as $\Lrw \charfct = 0$, the second statement follows from (2). \\
Part (5): The statement about $\Lsym$ follows from (1),
and then the statement about $\Lrw$ follows from (2). 
\end{proof}

As it is the case for the unnormalized graph Laplacian, the
multiplicity of the eigenvalue 0 of the normalized graph Laplacian is
related to the number of connected components: 

\begin{proposition}[Number of connected components and spectra of $\Lsym$ and $\Lrw$]
  \label{prop-connected-norm} Let $G$ be an undirected graph with
  non-negative weights. Then the multiplicity $k$ of the eigenvalue
  $0$ of both $\Lrw$ and $\Lsym$ equals the number of connected
  components $A_1, \hdots, A_k$ in the graph. For $\Lrw$, the eigenspace of $0$ is
  spanned by the indicator vectors $\charfct_{A_i}$ of those
  components. For $\Lsym$, the eigenspace of $0$ is spanned by the
  vectors   $D^{1/2} \charfct_{A_i}$. 
\end{proposition}
\begin{proof} The proof is analogous to the one of Proposition
  \ref{prop-connected-unnorm}, using Proposition \ref{prop-norm}. 
\end{proof}
\section{Spectral Clustering Algorithms} \label{sec-algorithms}

Now we would like to state the most common spectral clustering
algorithms. For references and the history of spectral clustering
we refer to Section \ref{sec-extensions}. We assume that our data consists of $n$ ``points'' 
$x_1,\hdots,x_n$ which can be arbitrary objects. We measure their pairwise
similarities $s_{ij} = s(x_i,x_j)$ by some
similarity function which is symmetric and non-negative, and we denote the
corresponding similarity matrix by $S = (s_{ij})_{i,j=1 \hdots n}$.

\ulesalgorithm{Unnormalized spectral clustering}{
Input: Similarity matrix $S \in \R^{n \times n}$, number $k$ of clusters to con\-struct. 
\blobb{ Construct a similarity graph by one of the ways described in
   Section \ref{sec-graphs}. Let $W$ be its weighted adjacency matrix. }
\blobb{Compute the unnormalized Laplacian $\Lun$. }
\blobb{{\bf Compute the first  $\boldsymbol k$ eigenvectors $\boldsymbol{u_1,\hdots,u_{k}}$ of $\boldsymbol{\Lun}$.} }
\blobb{ Let $U \in \R^{n \times k}$ be the matrix containing the vectors $u_1,\hdots,u_k$ as
columns. }
\blobb{ For $i=1,\hdots,n$, let $y_i \in \R^{k}$ be the vector corresponding to the $i$-th
row of $U$. }
\blobb{ Cluster the points $(y_i)_{i=1,\hdots,n}$ in $\R^k$ with the $k$-means
algorithm into clusters $C_1,\hdots,C_k$. }
Output: Clusters $A_1,\hdots,A_k$ with $A_i = \{j | \; y_j \in
C_i\}$. 
}

There are two different versions of normalized spectral clustering,
depending which of the normalized graph Laplacians is used. We name
both algorithms after two popular papers, for more references and
history please see Section~\ref{sec-extensions}.

\ulesalgorithm{Normalized spectral clustering according to
  \citeA{ShiMal00}}{%
Input: Similarity matrix $S \in \R^{n \times n}$, number $k$ of clusters to con\-struct. 
\blobb{ Construct a similarity graph by one of the ways described in
   Section \ref{sec-graphs}. Let $W$ be its weighted adjacency matrix. }
\blobb{Compute the unnormalized Laplacian $\Lun$. }
\blobb{{\bf Compute the first $\boldsymbol k$  generalized eigenvectors $\boldsymbol{ u_1,\hdots,u_{k}}$ of the
  generalized eigenproblem $\boldsymbol{\Lun u = \lambda D u}$. }}
\blobb{ Let $U\in \R^{n \times k}$ be the matrix containing the vectors $u_1,\hdots,u_k$ as
columns. }
\blobb{ For $i=1,\hdots,n$, let $y_i \in \R^{k}$ be the vector corresponding to the $i$-th
row of $U$. }
\blobb{ Cluster the points $(y_i)_{i=1,\hdots,n}$ in $\R^k$ with the $k$-means
algorithm into clusters $C_1, \hdots, C_k$. }
Output: Clusters $A_1,\hdots,A_k$ with $A_i = \{j | \; y_j \in
C_i\}$. 
}

Note that this algorithm uses the generalized eigenvectors of $\Lun$,
which according to Proposition \ref{prop-norm} correspond to the
eigenvectors of the matrix $\Lrw$. So in fact, the
algorithm works with eigenvectors of the normalized Laplacian $\Lrw$, and hence is called
normalized spectral clustering. 
The next algorithm also uses a normalized Laplacian, but this time the
matrix $\Lsym$ instead of $\Lrw$. As we will see, this algorithm needs
to introduce an additional row normalization step which is not needed
in the other algorithms. The reasons will
become clear in Section \ref{sec-perturbation}. \\

\ulesalgorithm{Normalized spectral clustering according to
  \citeA{NgJorWei02}}{%
Input: Similarity matrix $S \in \R^{n \times n}$, number $k$ of clusters to con\-struct. 
\blobb{ Construct a similarity graph by one of the ways described in
   Section \ref{sec-graphs}. Let $W$ be its weighted adjacency matrix. }
\blobb{Compute the normalized Laplacian $\Lsym$. }
\blobb{{\bf Compute the first $ \boldsymbol k$  eigenvectors
    $\boldsymbol{u_1,\hdots,u_{k}}$ of 
$\boldsymbol{\Lsym}$.} }
\blobb{ Let $U\in \R^{n \times k}$ be the matrix containing the vectors $u_1,\hdots,u_k$ as
columns. }
\blobb{ {\bf Form the matrix $T\in \R^{n \times k}$ from $U$ by normalizing the rows to norm 1}, \\
that is set $t_{ij} = u_{ij}/ (\sum_k u_{ik}^2)^{1/2}$. }
\blobb{ For $i=1,\hdots,n$, let $y_i \in \R^{k}$ be the vector corresponding to the $i$-th
row of $T$. }
\blobb{ Cluster the points $(y_i)_{i=1,\hdots,n}$ with the $k$-means
algorithm into clusters $C_1, \hdots, C_k$. }
Output: Clusters $A_1,\hdots,A_k$ with $A_i = \{j | \; y_j \in
C_i\}$. 
}

All three algorithms stated above look rather similar, apart from the
fact that they use three different graph Laplacians.  In all three
algorithms, the main trick is to change the representation of the
abstract data points $x_i$ to points $y_i \in \R^k$. It is due to the
properties of the graph Laplacians that this change of representation
is useful. We will see in the next sections that this change of
representation enhances the cluster-properties in the data, so that
clusters can be trivially detected in the new representation. In
particular, the simple $k$-means clustering algorithm has no
difficulties to detect the clusters in this new representation.
Readers not familiar with $k$-means can read up on this algorithm in numerous text books, for
example in \citeA{HasTibFri01}.  \\

\begin{figure}[bt!]
\includegraphics[height=0.10\textheight]{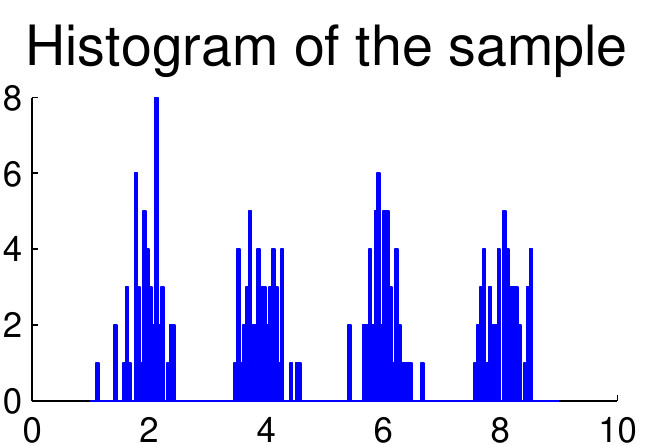}\\
\includegraphics[height=0.10\textheight]{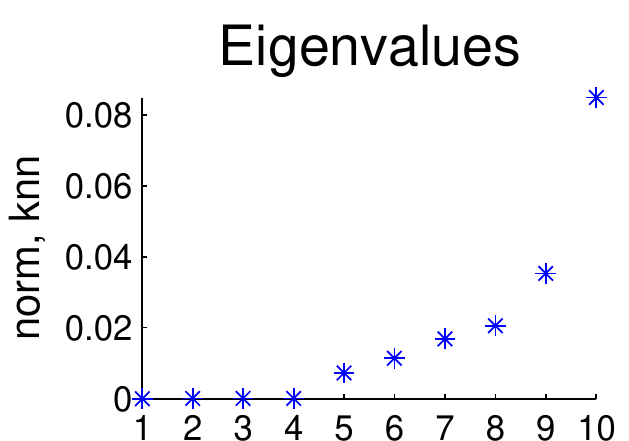}
\includegraphics[height=0.10\textheight]{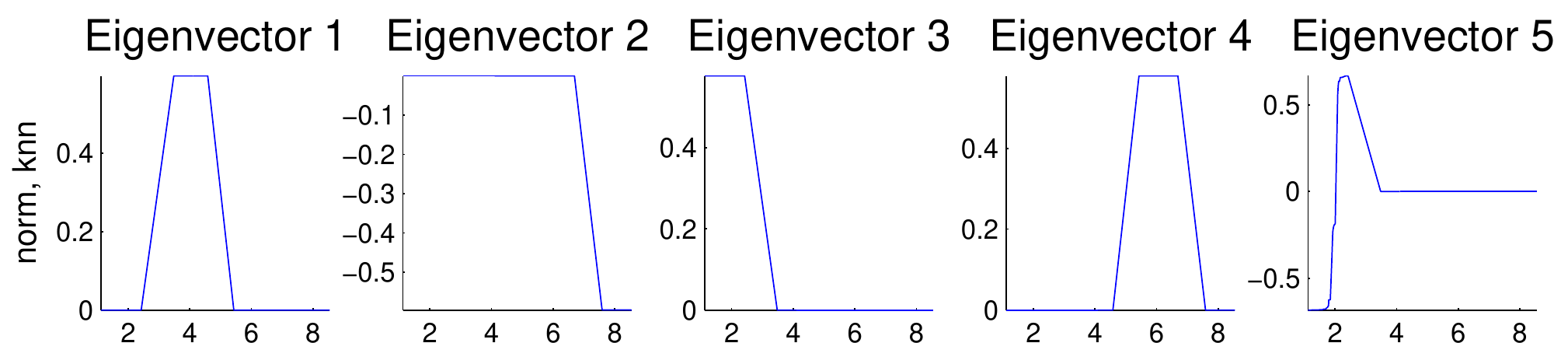}\\
\includegraphics[height=0.10\textheight]{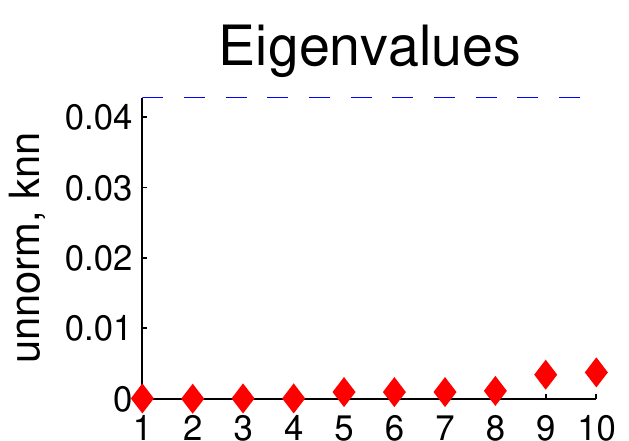}
\includegraphics[height=0.10\textheight]{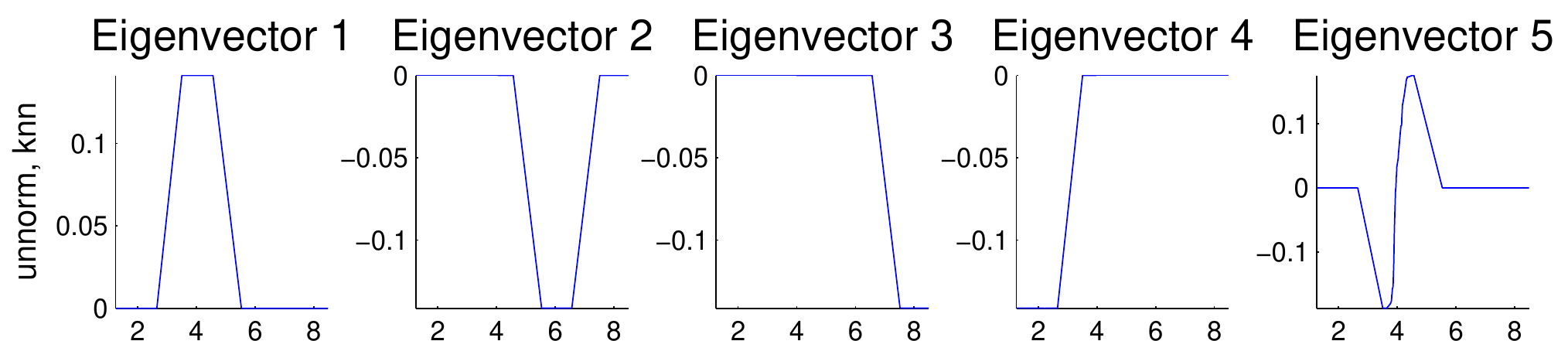}\\
\includegraphics[height=0.10\textheight]{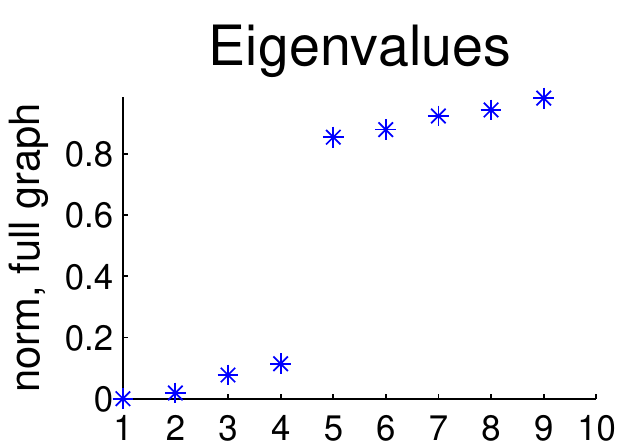}
\includegraphics[height=0.10\textheight]{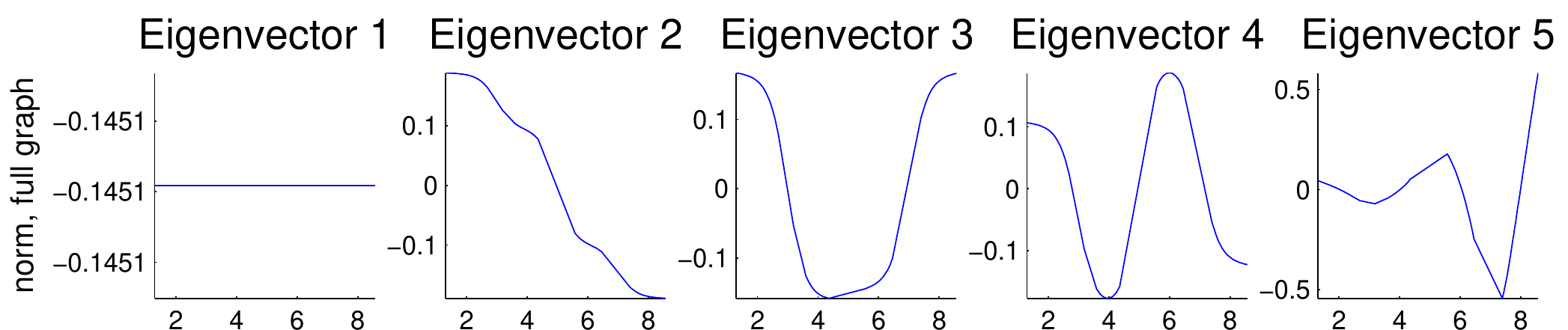}\\
\includegraphics[height=0.10\textheight]{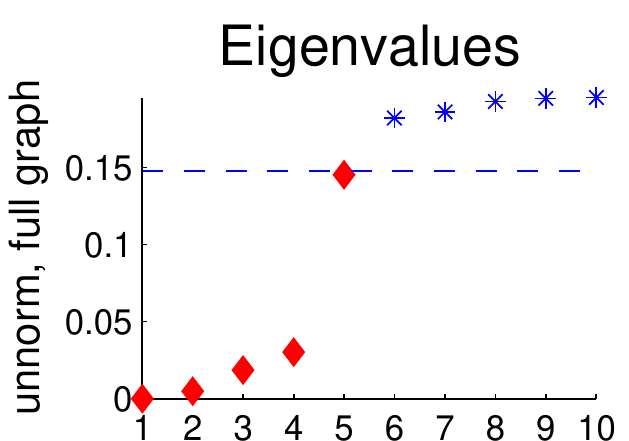}
\includegraphics[height=0.10\textheight]{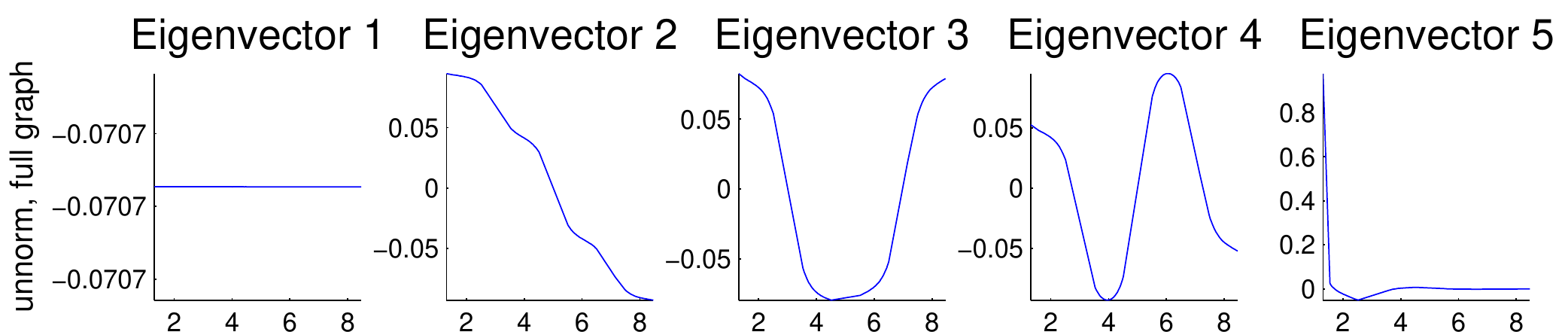}\\
\caption{Toy example for spectral clustering where the data points have been drawn from a mixture of four Gaussians on $\R$.  Left upper corner:
  histogram of the data. First and second row: eigenvalues and
  eigenvectors of $\Lrw$ and $\Lun$ based on the $k$-nearest neighbor
  graph. Third and fourth row: eigenvalues and eigenvectors of $\Lrw$
  and $\Lun$ based on the fully connected graph. For all plots, we
  used the Gaussian kernel with $\sigma=1$ as similarity
  function. See text for more details. }
\label{fig-example-spectral}
\end{figure}
Before we dive into the theory of spectral clustering, we would like
to illustrate its principle on a very simple toy example. This example
will be used at several places in this tutorial, and we chose it
because it is so simple that the relevant quantities can easily be
plotted. This toy data set consists of a random sample of 200 points
$x_1,\hdots,x_{200} \in \R$ drawn according to a mixture of four
Gaussians. The first row of Figure \ref{fig-example-spectral} shows
the histogram of a sample drawn from this distribution (the $x$-axis
represents the one-dimensional data space). As similarity function on this
data set we choose the Gaussian similarity function $s(x_i, x_j) =
\exp(-|x_i - x_j|^2/ (2\sigma^2))$ with $\sigma = 1$. As similarity graph
we consider both the fully connected graph and the $10$-nearest
neighbor graph. In Figure \ref{fig-example-spectral} we
show the first eigenvalues and eigenvectors of the unnormalized
Laplacian $\Lun$ and the normalized Laplacian $\Lrw$.  That is, in the
eigenvalue plot we plot $i$ vs. $\lambda_i$ (for the moment ignore the
dashed line and the different shapes of the eigenvalues in the plots
for the unnormalized case; their meaning will be discussed in Section
\ref{ssec-whichone}). In the eigenvector plots of an eigenvector $u =
(u_1,\hdots,u_{200})'$ we plot $x_i$ vs.  $u_i$ (note that in the
example chosen $x_i$ is simply a real number, hence we can depict it
on the $x$-axis).  The first two rows of Figure
\ref{fig-example-spectral} show the results based on the $10$-nearest
neighbor graph. We can see that the first four eigenvalues are 0, and
the corresponding eigenvectors are cluster indicator vectors. The
reason is that the clusters form disconnected parts in the $10$-nearest
neighbor graph, in which case the eigenvectors are given as in
Propositions \ref{prop-connected-unnorm} and
\ref{prop-connected-norm}. The next two rows show the results for the
fully connected graph. As the Gaussian similarity function is always
positive, this graph only consists of one connected component. Thus,
eigenvalue 0 has multiplicity 1, and the first eigenvector is the
constant vector. The following eigenvectors carry the information
about the clusters. For example in the unnormalized case (last row),
if we threshold the second eigenvector at 0, then the part below 0
corresponds to clusters 1 and 2, and the part above 0 to clusters 3
and 4. Similarly, thresholding the third eigenvector separates
clusters 1 and 4 from clusters 2 and 3, and thresholding the fourth
eigenvector separates clusters 1 and 3 from clusters 2 and 4.
Altogether, the first four eigenvectors carry all the information
about the four clusters. In all the cases illustrated in this figure,
spectral clustering using $k$-means on the first four eigenvectors
easily detects the correct four clusters.

\section{Graph cut point of view} \label{sec-graphcuts}

The intuition of clustering is to separate points in different groups
according to their similarities. For data given in form of a
similarity graph, this problem can be restated as follows: we want to
find a partition of the graph such that the edges between different
groups have a very low weight (which means that points in different
clusters are dissimilar from each other) and the edges within a group
have high weight (which means that points within the same cluster are
similar to each other). In this section we will see how spectral
clustering can be derived as an approximation to such graph
partitioning problems. \\

Given a similarity graph with adjacency matrix $W$, the simplest and
most direct way to construct a partition of the graph is to solve the
mincut problem. To define it, please recall the notation $W(A,B) :=
\sum_{i \in A, j \in B} w_{ij}$ and $\bar A$ for the complement of $A$. For a given number $k$ of subsets, the mincut approach simply consists
in choosing a partition $A_1,\hdots,A_k$ which minimizes 
\ba
\Cut(A_1,\hdots,A_k) := \frac{1}{2}\sum_{i=1}^k W(A_i, \bar A_i).
\ea
Here we introduce the factor $1/2$ for notational consistency, otherwise
we would count each edge twice in the cut. %
In particular for
$k=2$, mincut is a relatively easy problem and can be solved
efficiently, see \citeA{StoWag97} and the discussion therein.
However, in practice it often does not lead to satisfactory
partitions.  The problem is that in many cases, the solution of mincut
simply separates one individual vertex from the rest of the graph. Of
course this is not what we want to achieve in clustering, as clusters
should be reasonably large groups of points.
One way to circumvent this problem is to explicitly request that the
sets $A_1,\hdots, A_k$ are ``reasonably large''.  The two most common
objective functions to encode this are $\Ratiocut$ (\citeNP{HagKah92}) and the normalized cut $\Ncut$ (\citeNP{ShiMal00}).  In $\Ratiocut$, the size of a subset $A$ of a
graph is measured by its number of vertices $|A|$, while in $\Ncut$
the size is measured by the weights of its edges $\vol(A)$. The
definitions are:
\ba
& \Ratiocut(A_1,\hdots,A_k) 
:= \frac{1}{2}\sum_{i=1}^k \frac{W(A_i, \bar  A_i)}{|A_i|} 
= \sum_{i=1}^k \frac{\Cut(A_i, \bar  A_i)}{|A_i|}\\
& \Ncut(A_1,\hdots,A_k) 
:= \frac{1}{2}\sum_{i=1}^k \frac{W(A_i, \bar A_i)}{\vol(A_i)}
= \sum_{i=1}^k \frac{\Cut(A_i, \bar A_i)}{\vol(A_i)}.
\ea

Note that both objective functions take a small value if the clusters
$A_i$ are not too small. In particular, the minimum of the function
$\sum_{i=1}^k (1/|A_i|)$ is achieved if all $|A_i|$ coincide, and the
minimum of $\sum_{i=1}^k (1/\vol(A_i))$ is achieved if all $\vol(A_i)$
coincide. So what both objective functions try to achieve is that the
clusters are ``balanced'', as measured by the number of vertices or
edge weights, respectively.
Unfortunately, introducing balancing conditions makes the previously
simple to solve mincut problem become NP hard, see \citeA{WagWag93}
for a discussion. Spectral clustering is a way to solve relaxed
versions of those problems. We will see that relaxing $\Ncut$ leads to
normalized spectral clustering, while relaxing $\Ratiocut$ leads to
unnormalized spectral clustering (see also the tutorial slides by
\citeA{Ding04_tutorial}). 

  \subsection{Approximating $\boldsymbol{\Ratiocut}$ for $\boldsymbol{k=2}$ } 

Let us start with the case of $\Ratiocut$ and $k=2$, because the
relaxation is easiest to understand in this setting. Our goal is to
solve the optimization problem 
\banum \label{eq-obj-ratio}
\min_{A \subset V} \Ratiocut(A, \bar A).
\eanum
We first rewrite the problem in a more convenient
form. Given a subset $A \subset V$ we define the vector $f = (f_1,\hdots,f_n)' \in
  \RR^n$ with entries 
\banum \label{eq-def-f} 
f_i = 
\begin{cases} 
\sqrt{|\overline A| / |A|} & \text{ if } v_i \in A \\
- \sqrt{|A| / |\bar A|} & \text{ if } v_i \in \bar A. 
\end{cases}
\eanum
Now the $\Ratiocut$ objective function can be
conveniently rewritten using the unnormalized graph Laplacian. This is
due to the following calculation: 
\ba
f'\Lun f & = 
\frac{1}{2}
\sum_{i,j=1}^n w_{ij} (f_i - f_j)^2 \\
& = 
\frac{1}{2}
\sum_{i \in A, j \in \bar A}  w_{ij}
\left(
\sqrt{\frac{|\bar A|}{|A|}} + \sqrt{\frac{|A|}{|\bar A|}}
\right)^2 
+ 
\frac{1}{2}\sum_{i \in \bar A, j \in A}  w_{ij}
\left(
- \sqrt{\frac{|\bar A|}{|A|}} - \sqrt{\frac{|A|}{|\bar A|}}
\right)^2 \\
& = 
\Cut(A,\bar A)
\left(
\frac{|\bar A|}{|A|} + \frac{|A|}{|\bar A|} + 2
\right) \\
& =
\Cut(A,\bar A)
\left(
\frac{|A| + |\bar A|}{|A|} + \frac{|A| + |\bar A|}{|\bar A|}
\right) \\
& = 
|V| \cdot \Ratiocut(A, \bar A).
\ea
Additionally, we have 
\ba
\sum_{i=1}^n f_i = \sum_{i \in A} 
\sqrt{\frac{|\bar A|}{|A|}} - 
\sum_{i \in \bar A} \sqrt{\frac{|A|}{|\bar A|}}
= |A| \sqrt{\frac{|\bar A|}{|A|}} - |\bar A| \sqrt{\frac{|A|}{|\bar
    A|}} = 0.
\ea 

In other words, the vector $f$ as defined in Equation \eqref{eq-def-f}
is
orthogonal to the constant one vector $\charfct$. Finally, note that
$f$ satisfies 
\ba
\norm[f]^2 = \sum_{i=1}^n f_i^2 = |A| {\frac{|\bar A|}{|A|}} + |\bar A|
{\frac{|A|}{|\bar A|}} = |\bar A| + |A| = n.
\ea
Altogether we can see that the problem of minimizing \eqref{eq-obj-ratio} can be equivalently rewritten
as
\banum
\min_{A \subset V} f'Lf \text{ subject to } f \perp \charfct, \;
f_i \text{ as defined in Eq. \eqref{eq-def-f}}, \;
\norm[f] = \sqrt{n}.
\eanum

This is a discrete optimization problem as the entries of the solution
vector $f$ are only allowed to take two particular values, and of
course it is still NP hard.  The most obvious relaxation in this
setting is to discard the discreteness condition and instead allow
that $f_i$ takes arbitrary values in $\RR$. This leads to the relaxed
optimization problem
\banum
\min_{f \in \RR^n} f'Lf \text{ subject to } f \perp \charfct, \; \norm[f] = \sqrt{n}.
\eanum
By the Rayleigh-Ritz theorem (e.g., see Section 5.5.2. of \citeNP{Luetkepohl96})
it can be seen immediately that the solution of this problem is given
by the vector $f$ which is the eigenvector corresponding to the second
smallest eigenvalue of $\Lun$ (recall that the smallest eigenvalue of
$\Lun$ is 0 with eigenvector $\charfct$). So we can approximate a
minimizer of $\Ratiocut$ by the second eigenvector of
$\Lun$. However, in order to obtain a partition of the graph 
we need to re-transform the real-valued solution vector $f$ of the
relaxed problem 
into a discrete indicator vector.  The simplest way to do
this is to
use the sign of $f$ as indicator function, that is to choose 
\ba
\begin{cases}
v_i \in A & \text{ if } f_i \geq 0\\
v_i \in \bar A & \text{ if } f_i < 0.\\
\end{cases}
\ea 

However, in particular in the case of $k>2$ treated below, this
heuristic is too simple. What most spectral clustering algorithms do instead is to consider
the coordinates $f_i$ as points in $\R$ and cluster them into two
groups $C, \bar C$ by the $k$-means clustering algorithm. Then we
carry over the resulting clustering to the underlying data points, that is
we choose \ba
\begin{cases}
v_i \in A & \text{ if } f_i \in C\\
v_i \in \bar A & \text{ if } f_i \in \bar C.\\
\end{cases}
\ea
This is exactly the {\em unnormalized spectral clustering} algorithm
for the case
of $k=2$. \\

  \subsection{Approximating $\boldsymbol{\Ratiocut}$ for arbitrary $\boldsymbol{k}$} 

The relaxation of the $\Ratiocut$ minimization problem in the case of
a general value $k$ follows a similar principle as the one above. Given a partition
of $V$ into $k$ sets $A_1,\hdots,A_k$, we define $k$
indicator vectors $h_j = (h_{1,j},\hdots,h_{n,j})'$ by 
\banum
\label{eq-H}
h_{i,j} = 
\begin{cases} 1/\sqrt{ |A_j|} & \text{ if } v_i \in A_j \\
0 & \text{ otherwise}
\end{cases}
\hspace{2cm}
\; (i=1,\hdots,n; \; j = 1,\hdots, k).
\eanum
Then we set the matrix $H \in \R^{n \times k}$ as the matrix containing those $k$
indicator vectors as columns. Observe that the columns in $H$ are orthonormal to each
other, that is $H'H = I$. Similar to the calculations in the last
section we can see that 
\ba
h_i' \Lun h_i = \frac{\Cut(A_i, \bar A_i)}{|A_i|}.
\ea
Moreover, one can check that 
\ba
h_i' \Lun h_i = (H' \Lun H)_{ii}.
\ea
Combining those facts we get 
\ba
\Ratiocut(A_1,\hdots,A_k) =
\sum_{i=1}^k h_i' \Lun h_i 
=  \sum_{i=1}^k (H' \Lun H)_{ii} 
=  \Tr( H'\Lun H), 
\ea
where $\Tr$ denotes the trace of a matrix. 
So the problem of minimizing $\Ratiocut(A_1,\hdots,A_k)$ can be rewritten 
as 
\ba
\min_{A_1,\hdots,A_k} \Tr(H' \Lun H) \text{ subject to } H'H = I, \; H \text{ as defined
  in Eq. \eqref{eq-H}}.
\ea
Similar to above we now
relax the problem by allowing the entries of the matrix $H$ to take
arbitrary real values. Then the relaxed problem becomes: 
\ba
\min_{H \in \R^{n \times k}} \Tr (H' \Lun H)\; \text{ subject to } H'H = I.
\ea
This is the standard form of a trace minimization problem, and again
a version of the Rayleigh-Ritz theorem (e.g., see Section 5.2.2.(6)
of \citeNP{Luetkepohl96}) tells us that the solution is given by choosing
$H$ as the matrix which contains the first $k$ eigenvectors of $\Lun$
as columns. We can see that the  matrix $H$ is in fact the matrix $U$ used in the unnormalized
spectral clustering algorithm as described in Section
\ref{sec-algorithms}. Again we need to re-convert the real valued
solution matrix to a discrete partition. As above, the standard way is to
use the $k$-means algorithms on the rows
of $U$. This leads to the general unnormalized spectral
clustering algorithm as presented in Section \ref{sec-algorithms}.

  \subsection{Approximating  $\boldsymbol{\Ncut}$} 

Techniques very similar to the ones used for $\Ratiocut$  can be used to
derive normalized spectral clustering as relaxation of 
minimizing $\Ncut$. 
In the case $k=2$ we define the cluster indicator vector $f$ by 
\banum
\label{eq-def-f-norm}
f_i = 
\begin{cases}
\sqrt{\frac{\vol(\bar A)}{\vol A}} & \text{ if } v_i \in A \\
- \sqrt{\frac{\vol(A)}{\vol(\bar A)}} & \text{ if }  v_i \in \bar A.
\end{cases}
\eanum
Similar to above one can check that $(D f)'
\charfct = 0$, $f'D f = \vol(V)$, and $f' \Lun f = \vol(V)
\Ncut(A,\bar A)$. Thus we can rewrite the problem of minimizing
$\Ncut$ by the equivalent problem
\banum
\label{eq-norm-orig}
\min_{A} {f' \Lun f} \; \text{ subject to } \; 
f \text{ as in (\ref{eq-def-f-norm})},  \; 
D f \perp \charfct, \; 
f'Df = \vol(V).
\eanum
Again we relax the problem by allowing $f$ to take arbitrary real values: 
\banum
\label{eq-norm-relaxed}
\min_{f \in \R^n} {f' \Lun f} \; \text{ subject to } \; 
D f \perp \charfct, \; 
f'Df = \vol(V).
\eanum
Now we substitute $g := D^{1/2}f$. After substitution, the problem is 
\banum
\label{eq-norm-substituted}
\min_{g \in \R^n} g' D^{-1/2}\Lun D^{-1/2} g \; \text{ subject to } \; 
g  \perp D^{1/2}\charfct, \; 
\norm[g]^2 = \vol(V).
\eanum
Observe that $D^{-1/2}\Lun D^{-1/2} = \Lsym$, $D^{1/2}\charfct$ is
the first eigenvector of $\Lsym$, and $\vol(V)$ is a constant. Hence, 
Problem \eqref{eq-norm-substituted} is in the form
of the standard Rayleigh-Ritz theorem, and its solution $g$ is given
by the second eigenvector of $\Lsym$. Re-substituting $f = D^{-1/2} g$
and using Proposition \ref{prop-norm} we see that $f$ is the
second eigenvector of $\Lrw$, or equivalently the generalized
eigenvector of $\Lun u = \lambda D u$. \\

For the case of finding $k>2$ clusters, we define the indicator
vectors $h_j = (h_{1,j},\hdots,h_{n,j})'$ by 
\banum
\label{eq-norm-def-h}
h_{i,j} = 
\begin{cases} 1/\sqrt{ \vol(A_j)} & \text{ if } v_i \in A_j \\
0 & \text{ otherwise}
\end{cases}
\hspace{2cm}
\; (i=1,\hdots,n; \; j = 1,\hdots, k).
\eanum
Then we set the matrix $H$ as the matrix containing those $k$
indicator vectors as columns. Observe that $H'H = I$, $h_i'Dh_i = 1$,
and $h_i'L h_i = \Cut(A_i, \bar A_i) / \vol(A_i)$. So we can write
the problem of minimizing $\Ncut$ as 
\ba
\min_{A_1,\hdots,A_k} \Tr(H'LH) \text{ subject to } 
H'DH = I, \; 
\text{ H as in \eqref{eq-norm-def-h} }.
\ea
Relaxing the discreteness condition and substituting $T = D^{1/2}H$ we
obtain the relaxed problem 
\banum
\label{eq-norm-subst}
\min_{T \in \R^{n\times k}} \Tr(T'D^{-1/2}\Lun D^{-1/2}T) 
\text{ subject to } T'T = I.
\eanum
Again this is the standard trace minimization problem which is solved
by the matrix $T$ which contains the first $k$ eigenvectors of $\Lsym$
as columns. Re-substituting $H = D^{-1/2}T$ and using Proposition
\ref{prop-norm} we see that the solution $H$ consists of the first $k$
eigenvectors of the matrix $\Lrw$, or the first $k$ generalized
eigenvectors of $\Lun u = \lambda D u$. This yields the normalized
spectral clustering algorithm according to \citeA{ShiMal00}.\\

  \subsection{Comments on the relaxation approach} 

  There are several comments we should make about this derivation of
  spectral clustering. Most importantly, there is no guarantee
  whatsoever on the quality of the solution of the relaxed problem
  compared to the exact solution. That is, if $A_1,\hdots,A_k$ is the
  exact solution of minimizing $\Ratiocut$, and $B_1, \hdots, B_k$ is
  the solution constructed by unnormalized spectral clustering, then
  $\Ratiocut(B_1, \hdots, B_k) - \Ratiocut(A_1,\hdots,A_k)$ can be
  arbitrary large. Several examples for this can be found in
  \citeA{GuaMil98}. For instance, the authors consider a very simple class of
  graphs called ``cockroach graphs''.  Those graphs essentially look
  like a ladder, with a few rimes removed, see Figure
  \ref{fig-cockroach}.
\begin{figure}[tb]
\begin{center}
\includegraphics[width=0.5\textwidth]{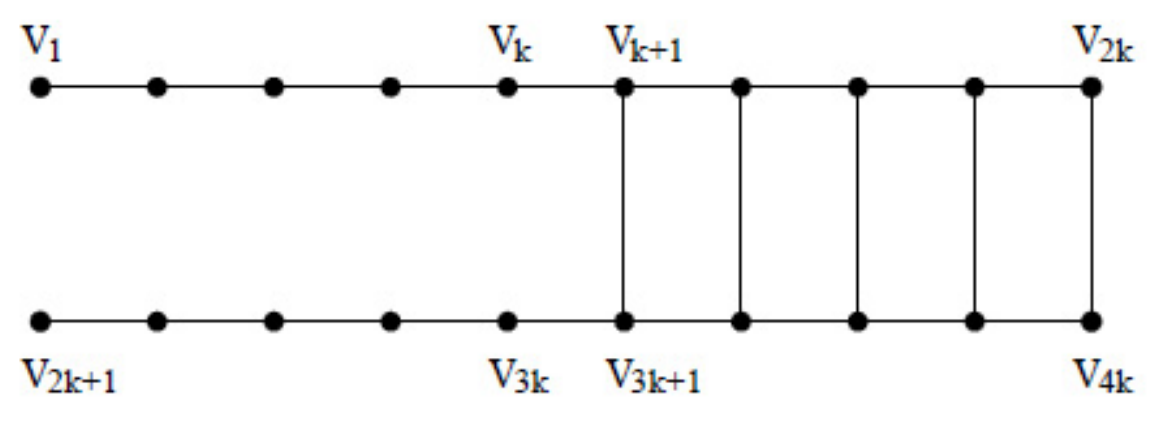}
\end{center}
\caption{The cockroach graph from \protect\citeA{GuaMil98}.}
\label{fig-cockroach}
\end{figure} 
Obviously, the ideal $\Ratiocut$ for $k=2$ just cuts the ladder by a vertical
cut such that $A = \{ v_1,\hdots, v_k, v_{2k+1}, \hdots, v_{3k}\}$ and
$\bar A = \{v_{k+1},\hdots,v_{2k}, v_{3k+1},\hdots,v_{4k}\}$.  This
cut is perfectly balanced with $|A| = |\bar A| = 2k$ and $\Cut(A,\bar
A) = 2$. However, by studying the properties of the second eigenvector
of the unnormalized graph Laplacian of cockroach graphs the authors
prove that unnormalized spectral clustering always cuts horizontally
through the ladder, constructing the sets $B = \{v_1, \hdots, v_{2k}
\}$ and $\bar B = \{v_{2k+1},\hdots,v_{4k}\}$.  This also results in a
balanced cut, but now we cut $k$ edges instead of just $2$. So
$\Ratiocut(A,\bar A) = 2/k$, while $\Ratiocut(B, \bar B) = 1$.  This
means that compared to the optimal cut, the
$\Ratiocut$ value obtained by spectral clustering is $k/2$
times worse, that is a factor in the order of $n$.  %
Several other papers investigate the quality of the clustering
constructed by spectral clustering, for example \citeA{SpiTen96} (for
unnormalized spectral clustering) and
\citeA{KanVemVet04_journal} (for normalized spectral clustering). 
In general it is known that efficient algorithms to approximate
balanced graph cuts up to a constant factor do not exist. To the
contrary, this approximation problem can be NP hard itself
\cite{BuiJon92}. \\

Of course, the relaxation we discussed above is not unique.  For
example, a completely different relaxation which leads to a
semi-definite program is derived in \citeA{DebCri06}, and there might
be many other useful relaxations. 
The reason why the spectral relaxation is so appealing is not that it
leads to particularly good solutions. Its popularity is mainly due to the fact that it results in 
a standard linear algebra problem which is simple to solve. \\

\section{Random walks point of view} \label{sec-rw}

Another line of argument to explain spectral clustering is based on
random walks on the similarity graph. A random walk on a graph is a
stochastic process which randomly jumps from vertex to vertex.  We
will see below that spectral clustering can be interpreted as trying
to find a partition of the graph such that the random walk stays long
within the same cluster and seldom jumps between clusters. Intuitively
this makes sense, in particular together with the graph cut
explanation of the last section: a balanced partition with a low cut
will also have the property that the random walk does not have many
opportunities to jump between clusters. For background reading on random walks in general we
refer to \citeA{Norris97} and \citeA{Bremaud99}, and for random walks on
graphs we recommend \citeA{AldFil} and \citeA{Lovasz93}. 
Formally, the transition probability of jumping in one step from vertex $v_i$ to
vertex $v_j$ is proportional to the edge weight $w_{ij}$ and  is
given by $p_{ij}:= w_{ij}/d_i$. The transition matrix $P =
(p_{ij})_{i,j=1,\hdots,n}$ of the random
walk is thus defined by
\ba
P = D^{-1} W.
\ea
If the graph is connected and non-bipartite,
then the random walk always possesses a unique stationary distribution
$\pi = (\pi_1, \hdots, \pi_n)'$, where $\pi_i =
d_i/\vol(V)$.  
Obviously there is a tight relationship between $\Lrw$ and $P$, as
$\Lrw = I - P$.  As a consequence, $\lambda$ is an eigenvalue of
$\Lrw$ with eigenvector $u$ if and only if $1 - \lambda$ is an
eigenvalue of $P$ with eigenvector $u$.  It is well known that many
properties of a graph can be expressed in terms of the corresponding
random walk transition matrix $P$, see \citeA{Lovasz93} for an
overview. From this point of view it does not come as a surprise that
the largest eigenvectors of $P$ and the smallest eigenvectors of
$\Lrw$ can be used to describe cluster properties of the graph.

\subsubsection*{Random walks and $\boldsymbol{\Ncut}$}

A formal equivalence between $\Ncut$ and transition probabilities of
the random walk has been observed in \citeA{MeiShi01}.

\begin{proposition}[$\Ncut$ via transition probabilities]
  \label{prop-meila} 
Let $G$ be connected and non bi-partite.  Assume
  that we run the random walk $(X_t)_{t \in \Nat}$ starting with $X_0$ in the
  stationary distribution $\pi$. For disjoint subsets $A, B \subset
  V$, denote by $P(B |A) := P(X_1 \in B | X_0 \in A)$. Then:
\ba
\Ncut(A,\bar A) = P(\bar A |A) + P(A|\bar A).
\ea
\end{proposition}
\begin{proof}
First of all observe that 
\ba 
& P(X_0 \in A, X_1 \in B)
= \sum_{i \in A, j \in B} P(X_0 = i, X_1 = j) 
\; = \sum_{i \in A, j \in B} \pi_i p_{ij} \\
& = \sum_{i \in A, j \in B} \frac{d_i}{\vol(V)} \frac{w_{ij}}{d_i}
\; = \frac{1}{\vol(V)} \sum_{i \in A, j \in B} w_{ij}.
\ea
Using this we obtain 
\ba
& P(X_1 \in B | X_0 \in A)
= \frac{P(X_0 \in A, X_1 \in B)}{P(X_0 \in A)} \\
 & = 
\left( \frac{1}{\vol(V)} \sum_{i \in A, j \in B} w_{ij} \right) 
\left(\frac{\vol(A)}{\vol(V)} \right)^{-1}
= 
\frac{\sum_{i \in A, j \in B} w_{ij}}{\vol(A)}.
\ea
Now the proposition follows directly with the definition of $\Ncut$. 
\end{proof}

This proposition leads to a nice interpretation of $\Ncut$, and
hence of normalized spectral clustering. 
It tells us that when minimizing $\Ncut$, we actually look for a cut
through the graph such that a random
walk seldom transitions from $A$ to $\bar A$ and vice versa.

\subsubsection*{The commute distance}

A second connection between random walks and graph Laplacians
can be made via the commute distance on the graph.  The commute
distance (also called resistance distance) $c_{ij}$ between two
vertices $v_i$ and $v_j$ is the expected time it takes the random walk to
travel from vertex $v_i$ to vertex $v_j$ and back
\cite{Lovasz93,AldFil}. The commute distance has several nice
properties which make it particularly appealing for machine
learning. As opposed to the shortest path distance on a graph, the
commute distance between two vertices decreases if there are many
different short ways to get from vertex $v_i$ to vertex $v_j$. So instead
of just looking for the one shortest path, the commute distance looks
at the set of short paths. Points which are connected by a short path
in the graph and lie in the same high-density region of the graph are
considered closer to each other than points which are connected by a
short path but lie in different high-density regions of the graph. In
this sense, the commute distance seems particularly well-suited to be
used for clustering purposes.  \\

Remarkably, the commute distance on a graph can be computed with the
help of the generalized inverse (also called pseudo-inverse or
Moore-Penrose inverse) $L^\dagger$ of the graph Laplacian $\Lun$.  In
the following we denote $e_i = (0,\hdots0,1,0,\hdots,0)'$ as the $i$-th unit
vector.
To define the generalized inverse of
$\Lun$, recall that by Proposition \ref{prop-unnorm} the matrix $L$
can be decomposed as $L = U \Lambda U'$ where $U$ is the matrix
containing all eigenvectors as columns and $\Lambda$ the diagonal
matrix with the eigenvalues $\lambda_1,\hdots, \lambda_n$ on the
diagonal. As at least one of the eigenvalues is 0, the matrix $L$ is
not invertible. Instead, we define its generalized inverse as
$L^\dagger := U \Lambda^\dagger U'$ where the matrix
$\Lambda^\dagger$ is the diagonal matrix with diagonal entries
$1/\lambda_i$ if $\lambda_i \neq 0$ and $0$ if $\lambda_i = 0$. The
entries of $\Lun^\dagger$ can be computed as $l_ {ij}^\dagger =
\sum_{k=2}^n \frac{1}{\lambda_k} u_{ik}u_{jk}$. The matrix  $\Lun^\dagger$ is positive
semi-definite and symmetric. For further properties of $\Lun^\dagger$ see
\citeA{GutXia04}.

\begin{proposition}[Commute distance] \label{prop-commute}
Let $G=(V,E)$ a connected, undirected graph. Denote by $c_{ij}$
the commute distance between vertex $v_i$ and vertex $v_j$, and by
$\Lun^\dagger = (l_{ij}^\dagger)_{i,j=1,\hdots,n}$ the generalized
inverse of $\Lun$. Then we have: 
\ba
c_{ij} = \vol(V) (l_{ii}^\dagger - 2 l_{ij}^\dagger + l_{jj}^\dagger) 
= \vol(V) (e_i - e_j)' \Lun^\dagger (e_i - e_j).
\ea
\end{proposition}

This result has been published by \citeA{KleRan93}, where it has been
proved by methods of electrical network theory. For a proof using
first step analysis for random walks see \citeA{FouPirRenSae07}. 
There also exist other ways to express the commute distance with the
help of graph Laplacians. For example a method in terms of
eigenvectors of the normalized Laplacian $\Lsym$ can be found as
Corollary 3.2 in \citeA{Lovasz93}, and a method computing the commute
distance with the help of determinants of certain sub-matrices of
$\Lun$ can be found in \citeA{BapGutXia03}.\\

Proposition \ref{prop-commute} has an important consequence. It shows
that $\sqrt{c_{ij}}$ can be considered as a Euclidean distance
function on the vertices of the graph. This means that we can
construct an embedding which maps the vertices $v_i$ of the graph on
points $z_i \in \R^n$ such that the Euclidean distances between the
points $z_i$ coincide with the commute distances on the graph. This
works as follows. As the matrix $\Lun^\dagger$ is positive
semi-definite and symmetric, it induces an inner product on $\R^n$ (or
to be more formal, it induces an inner product on the subspace of
$\R^n$ which is perpendicular to the vector $\charfct$).  Now choose
$z_i$ as the point in $\R^n$ corresponding to the $i$-th row of the
matrix $U (\Lambda^\dagger)^{1/2}$. Then, by Proposition
\ref{prop-commute} and by the construction of $\Lun^\dagger$ we have
that $\langle z_i, z_j \rangle = e_i'\Lun^\dagger
e_j$ and $c_{ij} = \vol(V) ||z_i - z_j||^2$. \\

The embedding used in unnormalized spectral clustering is related to
the commute time embedding, but not identical. In spectral clustering,
we map the vertices of the graph on the rows $y_i$ of the matrix $U$,
while the commute time embedding maps the vertices on the rows $z_i$
of the matrix $ (\Lambda^\dagger)^{1/2} U $. That is, compared to the
entries of $y_i$, the entries of $z_i$ are additionally scaled by the
inverse eigenvalues of $\Lun$. Moreover, in spectral clustering we
only take the first $k$ columns of the matrix, while the commute time
embedding takes all columns.  Several authors now try to justify why
$y_i$ and $z_i$ are not so different after all and state a bit
hand-waiving that the fact that spectral clustering constructs
clusters based on the Euclidean distances between the $y_i$ can be
interpreted as building clusters of the vertices in the graph based on
the commute distance.
However, note that both approaches can differ considerably. For example, in the optimal case
where the graph consists of
$k$  disconnected components, the first $k$
eigenvalues of $L$ are 0 according to Proposition
\ref{prop-connected-unnorm}, and the first $k$ columns of $U$ consist
of the cluster indicator vectors. However, the first $k$ columns of
the matrix $ (\Lambda^\dagger)^{1/2} U $ consist of zeros only, as the
first $k$ diagonal elements of $\Lambda^\dagger$ are $0$. In this
case, the information contained in the first $k$ columns of $U$ is
completely ignored in the matrix $ (\Lambda^\dagger)^{1/2} U $, and
all the non-zero elements of the matrix $ (\Lambda^\dagger)^{1/2} U $
which can be found in columns $k+1$ to $n$ are not taken into account
in spectral clustering, which discards all those columns. 
On the other hand, those problems do not occur if the underlying graph
is connected. In this case, the only eigenvector with eigenvalue $0$
is the constant one vector, which can be ignored in both cases. The
eigenvectors corresponding to small eigenvalues $\lambda_i$ of $L$ are
then stressed in the matrix $ (\Lambda^\dagger)^{1/2} U $ as they are
multiplied by $\lambda_i^\dagger = 1/\lambda_i$. In such a situation,
it might be true that the commute time embedding and the spectral
embedding do similar things.\\

All in all, it seems that the commute time distance can be a helpful
intuition, but without making further assumptions there is only a
rather loose relation between spectral clustering and the commute
distance. It might be possible that those relations can be tightened,
for example if the similarity function is strictly positive
definite. However, we have not yet seen a precise mathematical
statement about this.

\section{Perturbation theory point of view} \label{sec-perturbation}

Perturbation theory studies the question of how eigenvalues and
eigenvectors of a matrix $A$ change if we add a small perturbation
$H$, that is we consider the perturbed matrix $\tilde A := A + H$.
Most perturbation theorems state that a certain distance between
eigenvalues or eigenvectors of $A$ and $\tilde A$ is bounded by a
constant times a norm of $H$. The constant usually depends on which
eigenvalue we are looking at, and how far this eigenvalue is separated
from the rest of
the spectrum (for a formal statement see below). 
The justification of spectral clustering is then the following: Let us
first consider the ``ideal case'' where the between-cluster similarity
is exactly 0. We have seen in Section \ref{sec-laplacians} that then
the first $k$ eigenvectors of $\Lun$ or $\Lrw$ are the indicator
vectors of the clusters. In this case, the points $y_i \in \R^k$
constructed in the spectral clustering algorithms have the form
$(0,\hdots,0,1,0,\hdots0)'$ where the position of the 1 indicates the
connected component this point belongs to. In particular, all $y_i$
belonging to the same connected component coincide. The $k$-means
algorithm will trivially find the correct partition by placing a
center point on each of the points $(0,\hdots,0,1,0,\hdots0)' \in
\R^k$.
In a ``nearly ideal case'' where we still have distinct clusters, but
the between-cluster similarity is not exactly 0, we consider the
Laplacian matrices to be perturbed versions of the ones of the ideal
case. Perturbation theory then tells us that the eigenvectors will be
very close to the ideal indicator vectors. The points $y_i$ might not
completely coincide with $(0,\hdots,0,1,0,\hdots0)'$, but do so up to
some small error term. Hence, if the perturbations are not too large,
then $k$-means algorithm will still separate the groups from each
other. \\

\subsection{The formal perturbation argument} 

The formal basis for the perturbation approach to spectral clustering is the 
Davis-Kahan theorem from matrix perturbation theory. This theorem
bounds the difference between eigenspaces of symmetric matrices under
perturbations. We state those results for completeness, but for
background reading we refer to Section V of \citeA{SteSun90} and
Section VII.3 of \citeA{Bhatia97}.
In perturbation theory, distances between subspaces are usually
measured using ``canonical angles'' (also called ``principal
angles''). To define principal angles, let $\Vcal_1$ and $\Vcal_2$ be
two $p$-dimensional subspaces of $\R^d$, and $V_1$ and $V_2$ two
matrices such that their columns form orthonormal systems for
$\Vcal_1$ and $\Vcal_2$, respectively. Then the cosines $\cos
\Theta_i$ of the principal angles $\Theta_i$ are the singular values
of $V_1'V_2$.  For $p=1$, the so defined canonical angles coincide
with the normal definition of an angle. Canonical angles can also be
defined if $\Vcal_1$ and $\Vcal_2$ do not have the same dimension, see
Section V of \citeA{SteSun90}, Section VII.3 of \citeA{Bhatia97}, or
Section 12.4.3 of \citeA{GolVan96}. The matrix $\sin \mathbf
\Theta(\Vcal_1, \Vcal_2)$ will denote the diagonal matrix with the sine of the 
canonical angles on the diagonal.

\begin{theorem}[Davis-Kahan]
Let $A, H \in \R^{n \times n}$ be symmetric matrices, and let
$\norm[\cdot]$ be the Frobenius norm or the two-norm for
matrices, respectively. Consider $\tilde A := A + H$ as a perturbed version of $A$. 
Let $S_1 \subset \R$ be an interval. Denote by $\sigma_{S_1}(A)$ the set of
eigenvalues of $A$ which are contained in $S_1$, and by $V_1$ the
eigenspace corresponding to all those eigenvalues (more formally,
$V_1$ is the image of the spectral projection induced by
$\sigma_{S_1}(A)$). 
Denote by 
$\sigma_{S_1}(\tilde A)$ and $\tilde V_1$ the analogous quantities for $\tilde A$. 
Define the distance between $S_1$ and the spectrum of $A$ outside of
$S_1$ as 
\ba \delta = \min \{|\lambda - s|; \; \lambda \text{ eigenvalue of
}A, \; \lambda \not\in S_1, \; s \in S_1 \}.  \ea 
Then the distance $d(V_1,\tilde V_1):=
\norm[\sin \mathbf \Theta(V_1,\tilde V_1) ]$ between the two subspaces
 $V_1$ and $\tilde V_1$ is bounded by 
\ba
d(V_1, \tilde V_1) \leq \frac{\norm[H]}{\delta}. 
\ea
\end{theorem}

For a discussion and proofs of this theorem see for example Section V.3
of \citeA{SteSun90}.  Let us try to decrypt this theorem, for
simplicity in the case of the unnormalized Laplacian (for the
normalized Laplacian it works analogously). The matrix $A$ will
correspond to the graph Laplacian $L$ in the ideal case where the
graph has $k$ connected components. The matrix $\tilde A$ corresponds
to a perturbed case, where due to noise the $k$ components in the
graph are no longer completely disconnected, but they are only
connected by few edges with low weight. We denote the corresponding
graph Laplacian of this case by $\tilde \Lun$. For spectral clustering
we need to consider the first $k$ eigenvalues and eigenvectors of
$\tilde \Lun$. Denote the eigenvalues of $\Lun$ by $\lambda_1, \hdots
\lambda_n$ and the ones of the perturbed Laplacian $\tilde \Lun$ by
$\tilde \lambda_1,\hdots, \tilde \lambda_n$.  Choosing the interval
$S_1$ is now the crucial point. We want to choose it such that both
the first $k$ eigenvalues of $\tilde \Lun$ and the first $k$
eigenvalues of $\Lun$ are contained in $S_1$. This is easier the
smaller the perturbation $H = \Lun - \tilde \Lun$ and the larger the
eigengap $|\lambda_k - \lambda_{k+1}|$ is. If we manage to find such a
set, then the Davis-Kahan theorem tells us that the eigenspaces
corresponding to the first $k$ eigenvalues of the ideal matrix $\Lun$ and
the first $k$ eigenvalues of the perturbed matrix $\tilde \Lun$ are very
close to each other, that is their distance is bounded by
$\norm[H]/\delta$. Then, as the eigenvectors in the ideal case are
piecewise constant on the connected components, the same will
approximately be true in the perturbed case. How good
``approximately'' is depends on the norm of the perturbation
$\norm[H]$ and the distance $\delta$ between $S_1$ and the $(k+1)$st
eigenvector of $\Lun$.  If the set $S_1$ has been chosen as the
interval $[0,\lambda_k]$, then $\delta$ coincides with the spectral
gap $|\lambda_{k+1} - \lambda_k|$. We can see from the theorem that
the larger this eigengap is, the closer the eigenvectors of the ideal
case and the perturbed case are, and hence the better spectral
clustering works.  Below we will see that the size of the eigengap can
also be used in a different context as a quality criterion for
spectral clustering, namely when choosing the number $k$ of clusters
to construct.\\

If the perturbation $H$ is too large or the eigengap is too small, we
might not find a set $S_1$ such that both the first $k$ eigenvalues of
$\Lun$ and $\tilde \Lun$ are contained in $S_1$. In this case, we need
to make a compromise by choosing the set $S_1$ to contain the first
$k$ eigenvalues of $\Lun$, but maybe a few more or less eigenvalues
of $\tilde \Lun$. The statement of the theorem then becomes weaker in
the sense that either we do not compare the eigenspaces corresponding
to the first $k$ eigenvectors of $L$ and $\tilde L$, but the
eigenspaces corresponding to the first $k$ eigenvectors of $L$ and the
first $\tilde k$ eigenvectors of $\tilde L$ (where $\tilde k$ is the
number of eigenvalues of $\tilde L$ contained in $S_1$). Or, it
can happen that $\delta$ becomes so small that the bound on the
distance between $d(V_1, \tilde V_1)$ blows up so much that it becomes useless. 

\subsection{Comments about the perturbation approach}

A bit of caution is needed when using perturbation theory arguments to
justify clustering algorithms based on eigenvectors of matrices. In
general, {\em any} block diagonal symmetric matrix has the property that
there exists a basis of eigenvectors which are zero outside the
individual blocks and real-valued within the blocks. For example,
based on this argument several authors use the eigenvectors of the
similarity matrix $S$ or adjacency matrix $W$ to discover
clusters.
However, being block diagonal in the ideal case of completely
separated clusters can be considered as a necessary condition for a
successful use of eigenvectors, but not a sufficient one. 
At least two more properties should be satisfied: \\

First, we need to make sure that the {\em order} of the eigenvalues
and eigenvectors is meaningful. In case of the Laplacians this is
always true, as we know that any connected component possesses exactly
one eigenvector which has eigenvalue 0. Hence, if the graph has $k$
connected components and we take the first $k$ eigenvectors of the
Laplacian, then we know that we have exactly one eigenvector per
component. However, this might not be the case for other matrices such
as $S$ or $W$. For example, it could be the case that the two largest
eigenvalues of a block diagonal similarity matrix $S$ come from the
same block.  In such a situation, if we take the first $k$
eigenvectors of $S$, some blocks will be represented several times,
while there are other blocks which we will miss completely (unless we
take certain precautions). This is the reason why using the
eigenvectors of
$S$ or $W$ for clustering should be discouraged.\\

The second property is that in the ideal case, the entries of
the eigenvectors on the components should be ``safely bounded away''
from 0.  
Assume that an eigenvector on the first connected component has an
entry $u_{1,i} > 0$ at position $i$. In the ideal case, the fact
that this entry is non-zero indicates that the corresponding point $i$
belongs to the first cluster. The other way round, if a point $j$ does
not belong to cluster 1, then in the ideal case it should be the case
that $u_{1,j}=0$.  Now consider the same situation, but with perturbed
data. The perturbed eigenvector $\tilde u$ will
usually not have any non-zero component any more; but if the noise is
not too large, then perturbation theory tells us that the entries 
$\tilde u_{1,i}$ and $\tilde u_{1,j}$ are still ``close'' to their
original values $u_{1,i}$ and $u_{1,j}$. So both entries $\tilde
u_{1,i}$ and $\tilde u_{1,j}$ will take some small values, say
$\eps_1$ and $\eps_2$. In practice, if those values are very small it is unclear how we
should interpret this situation.  Either we believe that small entries
in $\tilde u$ indicate that the points do not belong to the first
cluster (which then misclassifies the first data point $i$), or we
think that the entries already indicate class membership and classify
both points to the first cluster (which misclassifies point $j$).\\

For both matrices $\Lun$ and $\Lrw$, the eigenvectors in the ideal
situation are indicator vectors, so the second problem described above
cannot occur.  However, this is not true for the matrix $\Lsym$, which
is used in the normalized spectral clustering algorithm of
\citeA{NgJorWei02}.  Even in the ideal case, the eigenvectors of this
matrix are given as $D^{1/2}\charfct_{A_i}$.  If the degrees of the
vertices differ a lot, and in particular if there are vertices which
have a very low degree, the corresponding entries in the eigenvectors
are very small.  To counteract the problem described above, the
row-normalization step in the algorithm of \citeA{NgJorWei02} comes into
play. In the ideal case, the matrix $U$ in the algorithm has exactly
one non-zero entry per row. After row-normalization, the matrix $T$ in
the algorithm of \citeA{NgJorWei02} then
consists of the cluster indicator vectors. Note however, that this might not always work out correctly in practice.  Assume that we have $\tilde
u_{i,1} = \eps_1$ and $\tilde u_{i,2} = \eps_2$.  If we now normalize
the $i$-th row of $U$, both $\eps_1$ and $\eps_2$ will be multiplied
by the factor of $1/\sqrt{\eps_1^2 + \eps_2^2}$ and become rather
large.  
We now run into a similar problem as described
above: both points are likely to be classified into the same cluster,
even though they belong to different clusters. 
This argument shows that spectral clustering using the matrix $\Lsym$
can be problematic if the eigenvectors contain particularly small
entries. On the other hand, note that such small entries in the eigenvectors
only occur if some of the vertices have a particularly low degrees (as
the eigenvectors of $\Lsym$ are given by $D^{1/2}\charfct_{A_i}$). One
could argue that in such a case, the data point should be considered
an outlier anyway, and then it does not really matter in which cluster the point will end up. \\

To summarize, the conclusion is that both unnormalized spectral clustering and
normalized spectral clustering with $\Lrw$ are well justified by the
perturbation theory approach. Normalized spectral clustering with
$\Lsym$ can also be justified by perturbation theory, but it should be
treated with more care if the graph contains vertices with
very low degrees. \\

\section{Practical details} \label{sec-inpractice}

In this section we will briefly discuss some of the issues which come
up when actually implementing spectral clustering. There are several
choices to be made and parameters to be set. However, the discussion
in this section is mainly meant to raise awareness about the general
problems which an occur. For thorough studies on the behavior of
spectral clustering for various real world tasks we refer to the
literature.

\subsection{Constructing the similarity graph}

Constructing the similarity graph for spectral clustering is not a
trivial task, and little is known on theoretical implications of the
various constructions.

\subsubsection*{The similarity function itself}

Before we can even think about constructing a similarity graph, we
need to define a similarity function on the data. As we are going to
construct a neighborhood graph later on, we need to make sure that the
local neighborhoods induced by this similarity function are
``meaningful''. This means that we need to be sure that points which
are considered to be ``very similar'' by the similarity function are
also closely related in the application the data comes from. For
example, when constructing a similarity function between text
documents it makes sense to check whether documents with a high
similarity score indeed belong to the same text category. The global
``long-range'' behavior of the similarity function is not so important
for spectral clustering --- it does not really matter whether two data
points have similarity score 0.01 or 0.001, say, as we will not
connect those two points in the similarity graph anyway. In the common
case where the data points live in the Euclidean space $\R^d$, a
reasonable default candidate is the Gaussian similarity function
$s(x_i,x_j) = \exp(- \norm[x_i - x_j]^2/(2\sigma^2))$ (but of course we
need to choose the parameter $\sigma$ here, see below). Ultimately,
the choice of the similarity function depends on the domain the data
comes from, and no general advice can be given.

\subsubsection*{Which type of similarity graph}

\begin{figure}[bt]
\begin{center}
\includegraphics[height=0.2\textheight]{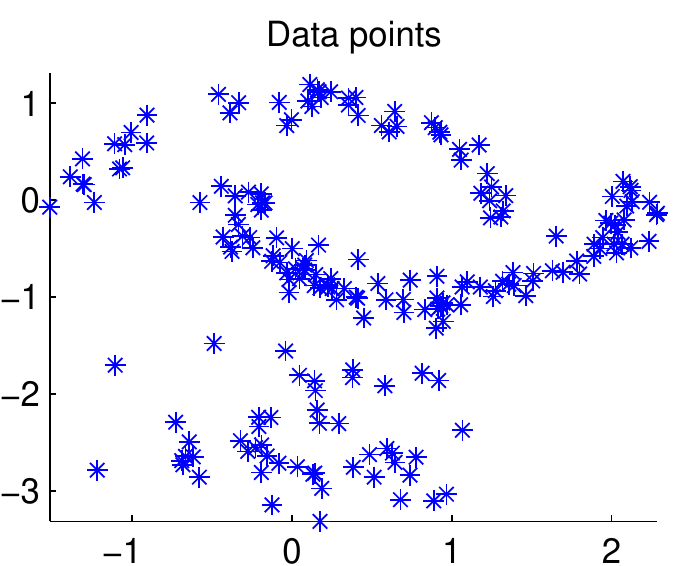}
\includegraphics[height=0.2\textheight]{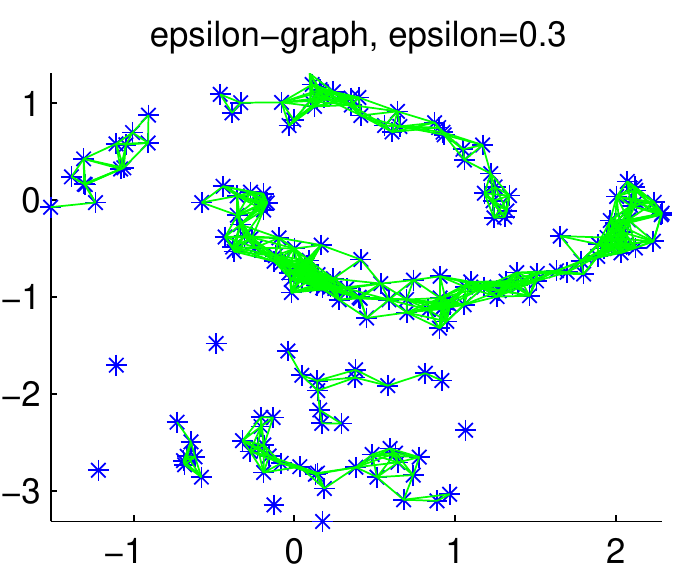}\\
\includegraphics[height=0.2\textheight]{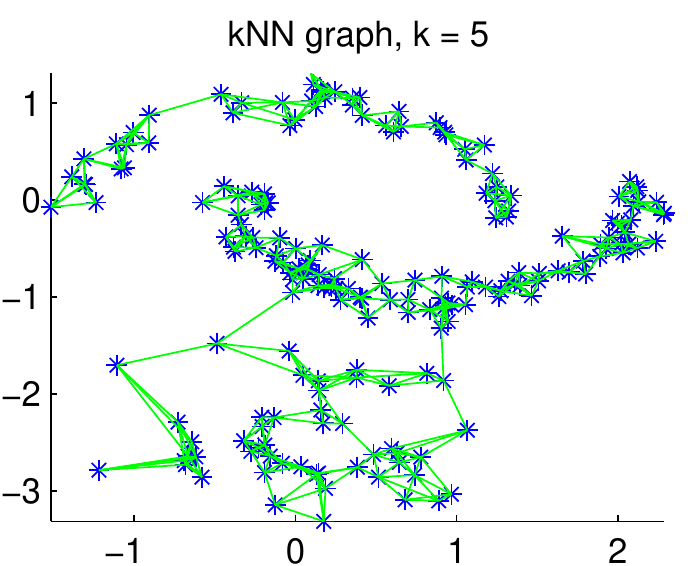}
\includegraphics[height=0.2\textheight]{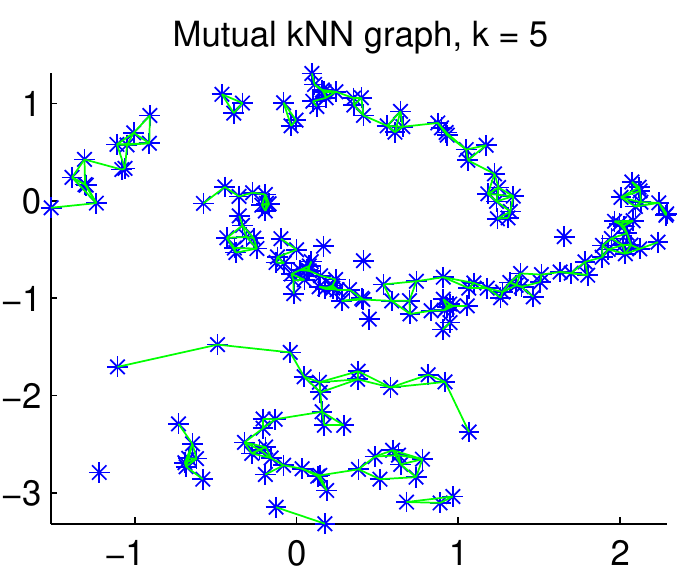}
\end{center}
\caption{Different similarity graphs, see text for details.}
\label{fig-similarity-graphs}
\end{figure}

The next choice one has to make concerns the type of the graph one
wants to use, such as the $k$-nearest neighbor or the
$\eps$-neighborhood graph. Let us illustrate the behavior of the
different graphs using the toy example presented in Figure
\ref{fig-similarity-graphs}. As underlying distribution we choose a
distribution on $\R^2$ with three clusters: two ``moons'' and a
Gaussian. The density of the bottom moon is chosen to be larger than
the one of the top moon. The upper left panel in Figure
\ref{fig-similarity-graphs} shows a sample drawn from this
distribution. The next three panels show the different similarity
graphs on this sample.\\

In the $\eps$-neighborhood graph, we can see that it is difficult to
choose a useful parameter $\eps$. With $\eps = 0.3$ as in the figure, the points on
the middle moon are already very tightly connected, while the points
in the Gaussian are barely connected. This problem always occurs if we
have data ``on different scales'', that is the distances between data
points are different in different regions of the space.\\

The $k$-nearest neighbor graph, on the other hand, can connect points
``on different scales''. We can see that points in the low-density
Gaussian are connected with points in the high-density moon. This is a
general property of $k$-nearest neighbor graphs which can be very
useful. We can also see that the $k$-nearest neighbor graph can break 
into several disconnected components if there are high density regions
which are reasonably far away from each other. This is the case for
the two moons in this example. \\

The mutual $k$-nearest neighbor graph has the property that it tends
to connect points within regions of constant density, but does not
connect regions of different densities with each other.  So the
mutual $k$-nearest neighbor graph can be considered as being ``in between''
the $\epsilon$-neighborhood graph and the $k$-nearest neighbor graph.
It is able to act on different scales, but does not mix those scales
with each other. Hence, the mutual $k$-nearest neighbor graph seems
particularly well-suited if we want to detect clusters of different
densities. \\

The fully connected graph is very often used in connection with the
Gaussian similarity function $s(x_i,x_j) = \exp(-
\norm[x_i-x_j]^2/ (2\sigma^2))$. Here the parameter $\sigma$ plays a
similar role as the parameter $\eps$ in the $\eps$-neighborhood graph.
Points in local neighborhoods are connected with relatively high
weights, while edges between far away points have positive, but
negligible weights. However, the resulting similarity matrix is not a
sparse matrix.\\

As a general recommendation we suggest to work with the $k$-nearest
neighbor graph as the first choice. It is simple to work with, results
in a sparse adjacency matrix $W$, and in our experience is less
vulnerable to unsuitable choices of parameters than the other graphs.

\subsubsection*{The parameters of the similarity graph} 

Once one has decided for the type of the similarity graph, one has to
choose its connectivity parameter $k$ or $\epsilon$, respectively.
Unfortunately, barely any theoretical results are known to guide us in
this task. 
In general, if the similarity graph contains more connected components
than the number of clusters we ask the algorithm to detect, then spectral
clustering will trivially return connected components as
clusters. Unless one is perfectly sure that those connected components
are the correct clusters, one should make sure that the
similarity graph is connected, or only consists of ``few'' connected
components and very few or no isolated vertices. There are many
theoretical results on how connectivity of random graphs can be
achieved, but all those results only hold in the limit for the sample
size $n \to \infty$. For example, it is known that for $n$ data points
drawn i.i.d. from some underlying density with a connected support
in $\R^d$, the $k$-nearest neighbor graph and the mutual
$k$-nearest neighbor graph will be connected if we choose $k$ on the
order of $\log(n)$ (e.g., \citeNP{BriChaQuiYuk97}).  Similar arguments
show that the parameter $\eps$ in the $\eps$-neighborhood graph has to
be chosen as $(\log(n)/n )^d$ to guarantee connectivity in the limit
(\citeNP{Penrose99}).  
While being of
theoretical interest, all those results do not really help us for
choosing $k$ on a finite sample.\\

Now let us give some rules of thumb. When working with the $k$-nearest
neighbor graph, then the connectivity parameter should be chosen such
that the resulting graph is connected, or at least has significantly
fewer connected components than clusters we want to detect. %
For small or medium-sized graphs this can
be tried out "by foot". For very large graphs, a first approximation
could be to choose $k$ in the order of $\log(n)$, as suggested by the
asymptotic connectivity results. \\

For the mutual $k$-nearest neighbor graph, we have to admit that we
are a bit lost for rules of thumb. The advantage of the mutual
$k$-nearest neighbor graph compared to the standard $k$-nearest
neighbor graph is that it tends not to connect areas of different
density. While this can be good if there are clear clusters induced by
separate high-density areas, this can hurt in less obvious situations
as disconnected parts in the graph will always be chosen to be
clusters by spectral clustering. Very generally, one can observe that
the mutual $k$-nearest neighbor graph has much fewer edges than the
standard $k$-nearest neighbor graph for the same parameter $k$. This
suggests to choose $k$ significantly larger for the mutual $k$-nearest
neighbor graph than one would do for the standard $k$-nearest neighbor
graph. %
However, to take advantage of the property that the
mutual $k$-nearest neighbor graph does not connect regions of
different density, it would be necessary to allow for several ``meaningful''
disconnected parts of the graph. Unfortunately, we do not know of any general
heuristic to choose the parameter $k$ such that this can be achieved. \\

For the $\eps$-neighborhood graph, we suggest to choose $\eps$ such
that the resulting graph is safely connected. To determine the
smallest value of $\eps$ where the graph is connected is very simple:
one has to choose $\eps$ as the length of the longest edge in a
minimal spanning tree of the fully connected graph on the data
points. The latter can be determined easily by any minimal spanning
tree algorithm. However, note that when the data contains outliers
this heuristic will choose $\eps$ so large that even the outliers are
connected to the rest of the data. A similar effect happens when the
data contains several tight clusters which are very far apart from
each other. In both cases, $\eps$ will be chosen too large to reflect
the scale of the most important part of the data. \\

Finally, if one uses a fully connected graph together with a similarity
function which can be scaled itself, for example the Gaussian
similarity function, then the scale of the similarity function should
be chosen such that the resulting graph has similar properties as a
corresponding $k$-nearest neighbor or $\eps$-neighborhood graph would
have. One needs to make sure that for most data points the set of
neighbors with a similarity significantly larger than 0 is ``not too
small and not too large''. In particular, for the Gaussian similarity
function several rules of thumb are frequently used. For example, one
can choose $\sigma$ in the order of the mean distance of a point to
its $k$-th nearest neighbor, where $k$ is chosen similarly as above
(e.g., $k \sim \log(n) + 1 $ ). Another way is to determine $\eps$ by the
minimal spanning tree heuristic described above, and then choose
$\sigma = \eps$. But note that all those rules of thumb are
very ad-hoc, and depending on the given data at hand and its
distribution of inter-point distances they might not work at all.\\

In general, experience shows that spectral clustering can be quite
sensitive to changes in the similarity graph and to the choice of its
parameters.  Unfortunately, to our knowledge there has been no
systematic study which investigates the effects of the similarity
graph and its parameters on clustering and comes up with
well-justified rules of thumb. None of the recommendations above is
based on a firm theoretic ground. Finding rules which have a
theoretical justification should be considered an interesting and
important topic for future research. \\

\subsection{Computing the eigenvectors} \label{ssec-computation}
To implement spectral clustering in practice one has to compute the
first $k$ eigenvectors of a potentially large graph Laplace matrix.
Luckily, if we use the $k$-nearest neighbor graph or the
$\eps$-neighborhood graph, then all those matrices are sparse.
Efficient methods exist to compute the first eigenvectors of sparse
matrices, the most popular ones being the power method or Krylov
subspace methods such as the Lanczos method \cite{GolVan96}. The
speed of convergence of those algorithms depends on the size of the
eigengap (also called spectral gap) $\gamma_k
= |\lambda_k - \lambda_{k+1}|$. The larger this eigengap is, the
faster the algorithms computing the first $k$ eigenvectors converge. \\

Note that a general problem occurs if one of the eigenvalues under
consideration has multiplicity larger than one. For example, in the
ideal situation of $k$ disconnected clusters, the eigenvalue $0$ has
multiplicity $k$. As we have seen, in this case the eigenspace is
spanned by the $k$ cluster indicator vectors. But unfortunately, the
vectors computed by the numerical eigensolvers do not necessarily
converge to those particular vectors.  Instead they just converge to
some orthonormal basis of the eigenspace, and it usually depends on
implementation details to which basis exactly the algorithm converges.
But this is not so bad after all. Note that all vectors in the space
spanned by the cluster indicator vectors $\charfct_{A_i}$ have the
form $u = \sum_{i=1}^k a_i \charfct_{A_i}$ for some coefficients
$a_i$, that is, they are piecewise constant on the clusters. So the
vectors returned by the eigensolvers still encode the information
about the clusters, which can then be used by the $k$-means
algorithm to reconstruct the clusters. \\

\subsection{The number of clusters} 

Choosing the number $k$ of clusters is a general problem for all
clustering algorithms, and a variety of more or less successful
methods have been devised for this problem. In model-based clustering
settings there exist well-justified criteria to choose the number of
clusters from the data. Those criteria are usually based on the
log-likelihood of the data, which can then be treated in a frequentist
or Bayesian way, for examples see \citeA{FraRaf02}. In settings where
no or few assumptions on the underlying model are made, a large
variety of different indices can be used to pick the number of
clusters. Examples range from ad-hoc measures such as the ratio of
within-cluster and between-cluster similarities, over
information-theoretic criteria (\citeNP{StiBia04}), the gap statistic
(\citeNP{TibWalHas01}), to stability approaches (\citeNP{BenEliGuy02};
\citeNP{LanRotBraBuh04}; \citeNP{BenLuxPal06}).  Of course all those
methods can also be used for spectral clustering. Additionally, one
tool which is particularly designed for spectral clustering is the
eigengap heuristic, which can be used for all three graph
Laplacians. Here the goal is to choose the number $k$ such that all
eigenvalues $\lambda_1, \hdots, \lambda_k$ are very small, but
$\lambda_{k+1}$ is relatively large.
There are several justifications for this procedure. The first one is
based on perturbation theory, where we observe that in the ideal case
of $k$ completely disconnected clusters, the eigenvalue $0$ has
multiplicity $k$, and then there is a gap to the $(k+1)$th eigenvalue
$\lambda_{k+1} > 0$.  Other explanations can be given by spectral
graph theory. Here, many geometric invariants of the graph can be
expressed or bounded with the help of the first eigenvalues of the
graph Laplacian. In particular, the sizes of cuts are closely related
to the size of the first eigenvalues. For more details on this topic
we refer to \citeA{Bolla91}, \citeA{Mohar97} and \citeA{Chung97}. \\

We would like to illustrate the eigengap heuristic on our toy example
introduced in Section \ref{sec-algorithms}. For this purpose we
consider similar data sets as in Section \ref{sec-algorithms}, but to
vary the difficulty of clustering we consider the Gaussians with
increasing variance. The first row of Figure \ref{fig-choosing-k}
shows the histograms of the three samples. We construct the 10-nearest
neighbor graph as described in Section \ref{sec-algorithms}, and plot
the eigenvalues  of the normalized Laplacian $\Lrw$ on
the different samples (the results for the unnormalized Laplacian are
similar). The first data set consists of four well separated clusters,
and we can see that the first 4 eigenvalues are approximately 0. Then
there is a gap between the 4th and 5th eigenvalue, that is $|\lambda_5 -
\lambda_4|$ is relatively large.  According to the eigengap heuristic,
this gap indicates that the data set contains 4 clusters. The same
behavior can also be observed for the results of the fully connected
graph (already plotted in Figure \ref{fig-example-spectral}). So we
can see that the heuristic works well if the clusters in the data are
very well pronounced.
However, the more noisy or overlapping the clusters are, the less
effective is this heuristic.  We can see that for the second data set
where the clusters are more ``blurry'', there is still a gap between
the 4th and 5th eigenvalue, but it is not as clear to detect as in the
case before.  Finally, in the last data set, there is no well-defined
gap, the differences between all eigenvalues are approximately the
same. But on the other hand, the clusters in this data set overlap so
much that many non-parametric algorithms will have difficulties to
detect the clusters, unless they make strong assumptions on the underlying
model. In this particular example, even for a human looking at the
histogram it is not obvious what the correct number of clusters should
be.  This illustrates that, as most methods for choosing the number of
clusters, the eigengap heuristic usually works well if the data
contains very well pronounced clusters, but in ambiguous cases it also
returns ambiguous results.\\

\begin{figure}[bt]
\begin{center}
\includegraphics[height=0.14\textheight]{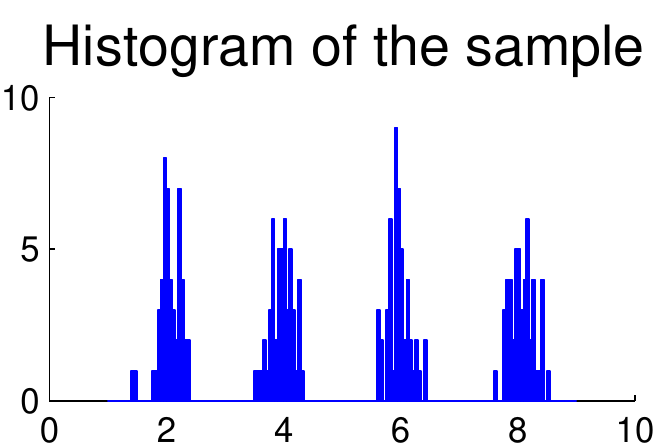}
\hfill
\includegraphics[height=0.14\textheight]{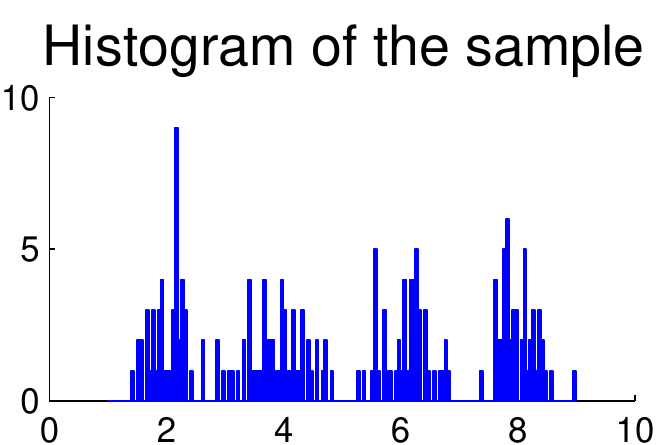}
\hfill
\includegraphics[height=0.14\textheight]{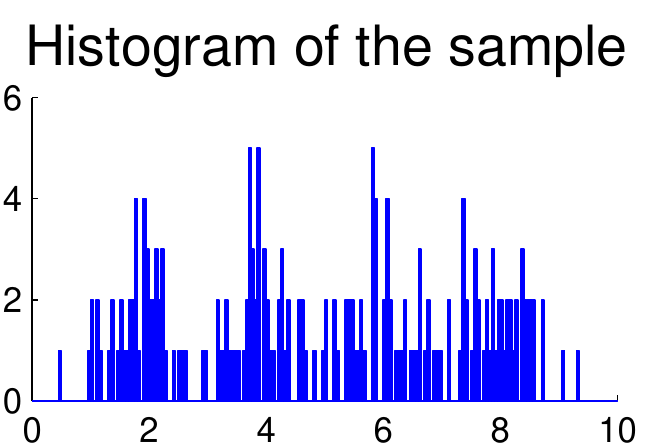}\\
\includegraphics[height=0.14\textheight]{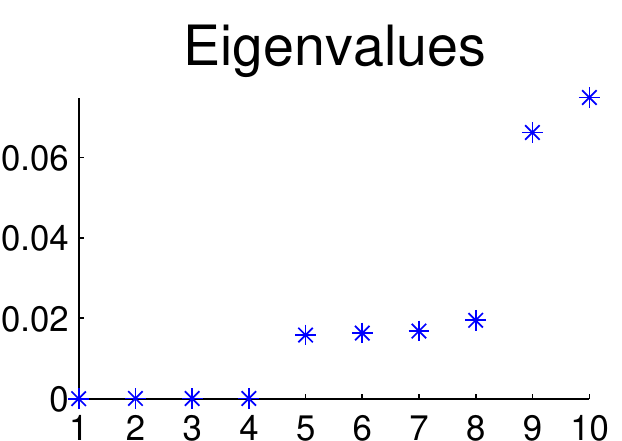}
\hfill
\includegraphics[height=0.14\textheight]{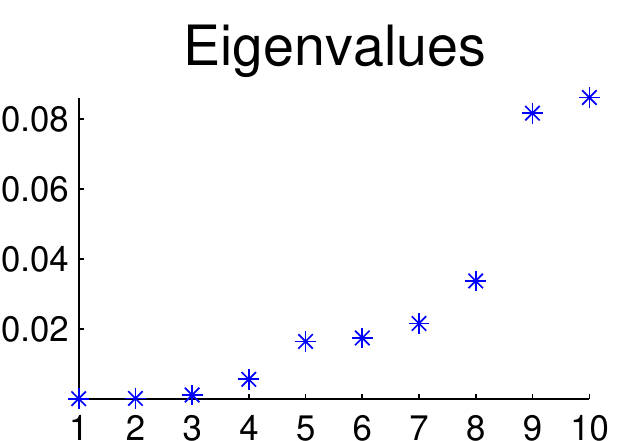}
\hfill
\includegraphics[height=0.14\textheight]{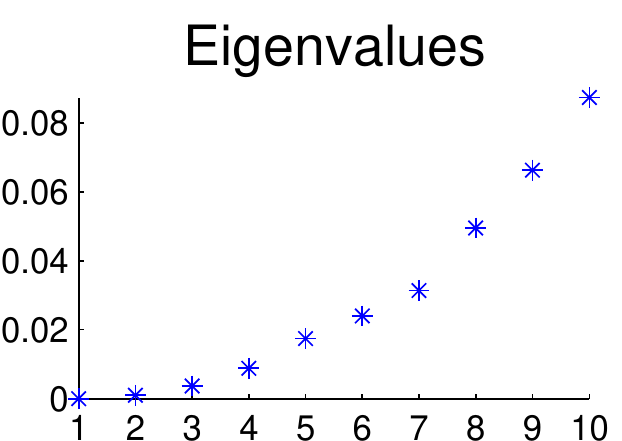}
\end{center}
\caption{Three data sets, and the smallest 10 eigenvalues of $\Lrw$.  See text for more details.}
\label{fig-choosing-k}
\end{figure}

Finally, note that the choice of the number of clusters and the choice
of the connectivity parameters of the neighborhood graph affect each
other.  For example, if the connectivity parameter of the neighborhood
graph is so small that the graph breaks into, say, $k_0$ connected
components, then choosing $k_0$ as the number of clusters is a valid
choice. However, as soon as the neighborhood graph is connected, it is
not clear how the number of clusters and the connectivity parameters
of the neighborhood graph interact. Both the choice of the number of
clusters and the choice of the connectivity parameters of the graph
are difficult problems on their own, and to our knowledge nothing
non-trivial is known on their interactions.

\subsection{The $k$-means step}

The three spectral clustering algorithms we presented in Section
\ref{sec-algorithms} use $k$-means as last step to extract the final
partition from the real valued matrix of eigenvectors. First of all,
note that there is nothing principled about using the $k$-means
algorithm in this step. In fact, as we have seen from the various
explanations of spectral clustering, this step should be very simple
if the data contains well-expressed clusters. For example, in the
ideal case if completely separated clusters we know that the
eigenvectors of $\Lun$ and $\Lrw$ are piecewise constant. In this
case, all points $x_i$ which belong to the same cluster $C_s$ are
mapped to exactly the sample point $y_i$, namely to the unit vector
$e_s \in \R^k$. In such a trivial case, any clustering algorithm
applied to the points $y_i \in \R^k$ will be able to extract the correct
clusters. \\

While it is somewhat arbitrary what clustering algorithm exactly one
chooses in the final step of spectral clustering, one can argue that
at least the Euclidean distance between the points $y_i$ is a
meaningful quantity to look at. We have seen that the Euclidean
distance between the points $y_i$ is related to the ``commute
distance'' on the graph, and in \citeA{NadLafCoi06} the authors show
that the Euclidean distances between the $y_i$ are also related to a
more general ``diffusion distance''. Also, other uses of the spectral
embeddings (e.g., \citeA{Bolla91} or \citeA{BelNiy03eigenmaps}) show
that the Euclidean distance in $\R^d$ is meaningful.\\

Instead of $k$-means, people also use other techniques to construct he
final solution from the real-valued representation. For example, in
\citeA{Lang06} the authors use hyperplanes for this purpose. A more
advanced post-processing of the eigenvectors is proposed in
\citeA{BachJordan04}. Here the authors study the subspace spanned by
the first $k$ eigenvectors, and try to approximate this subspace as
good as possible using piecewise constant vectors.  This also leads to
minimizing certain Euclidean distances in the space $\R^k$, which can
be done by some weighted $k$-means algorithm.\\

\subsection{Which graph Laplacian should be used? } \label{ssec-whichone}

A fundamental question related to spectral clustering is the question
which of the three graph Laplacians should be used to compute the
eigenvectors. Before deciding this
question, one should always look at the degree distribution of the similarity
graph. If the graph is very regular and most vertices have
approximately the same degree, then all the Laplacians are very
similar to each other, and will work equally well for clustering.
However, if the degrees in the graph are very broadly distributed, then the
Laplacians differ considerably. In our opinion, there are several
arguments which advocate for using normalized rather than unnormalized
spectral clustering, and in the normalized case to use
the eigenvectors of  $\Lrw$ rather than those of $\Lsym$. \\

\subsubsection*{Clustering objectives satisfied by the different algorithms}

The first argument in favor of normalized spectral clustering comes
from the graph partitioning point of view.  For simplicity let us
discuss the case $k=2$. In general, clustering has two different
objectives:

\begin{enumerate}
\item We want to find a partition such that points in different
  clusters are dissimilar to each other, that is we want to minimize
  the between-cluster similarity. In the graph setting, this
  means to minimize $\Cut(A,\bar A)$.
\item We want to find a partition such that points in the same cluster
  are similar to each other, that is we want to maximize the
  within-cluster similarities $W(A,A)$ and $W(\bar A, \bar A)$.
\end{enumerate}

Both $\Ratiocut$ and $\Ncut$ directly implement the first objective
by explicitly incorporating  $\Cut(A, \bar A)$ in the objective function.
However, concerning the second point, both algorithms behave
differently. Note that
\ba
W(A,A) = W(A,V) - W(A,\bar A)
= \vol(A) - \Cut(A, \bar A).
\ea
Hence, the within-cluster similarity is maximized if $\Cut(A, \bar A
)$ is small {\em and } if $\vol(A)$ is large. As this is exactly what
we achieve by minimizing $\Ncut$, the $\Ncut$ criterion
implements the second objective. This can be seen even more
explicitly by considering yet another graph cut objective function,
namely the $\MinmaxCut$ criterion introduced by \citeA{DingEtal01}:
\ba \MinmaxCut(A_1, \hdots,A_k):= \sum_{i=1}^k \frac{\Cut(A_i,\bar
  A_i)}{W(A_i,A_i)}.  \ea
Compared to $\Ncut$, which has the terms $\vol(A) = \Cut(A,\bar A) +
W(A,A)$ in the denominator, the $\MinmaxCut$ criterion only has
$W(A,A)$ in the denominator.  In practice, $\Ncut$ and $\MinmaxCut$
are often minimized by similar cuts, as a good $\Ncut$ solution will
have a small value of $\Cut(A,\bar A)$ anyway and hence the
denominators are not so different after all. Moreover, relaxing
$\MinmaxCut$ leads to exactly the same optimization problem as
relaxing $\Ncut$, namely to normalized
spectral clustering with the eigenvectors of $\Lrw$. So one can see by several ways that normalized spectral clustering incorporates both clustering objectives mentioned above.\\

Now consider the case of $\Ratiocut$. Here the objective is to maximize
$|A|$ and$|\bar A|$ instead of $\vol(A)$ and $\vol(\bar A)$. But $|A|$
and $|\bar A|$ are not necessarily related to the within-cluster similarity, as
the within-cluster similarity depends on the edges and not on the
number of vertices in $A$. For instance, just think of a set $A$ which has very many
vertices, all of which only have very low weighted edges to each other.
Minimizing $\Ratiocut$ does not attempt to maximize the
within-cluster similarity, and  the same is then true for its
relaxation by
unnormalized spectral clustering. \\

So this is our first important point to keep in mind: Normalized
spectral clustering implements both clustering objectives mentioned
above, while unnormalized spectral clustering only implements the
first objective.  \\

\subsubsection*{Consistency issues}

\begin{figure}[bt]
\begin{center}
\includegraphics[height=0.10\textheight]{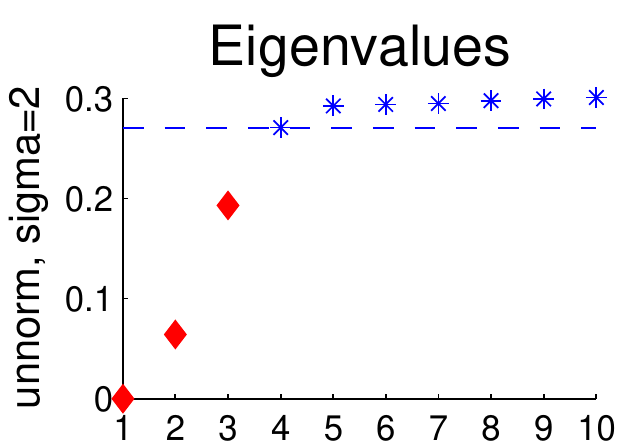}
\includegraphics[height=0.10\textheight]{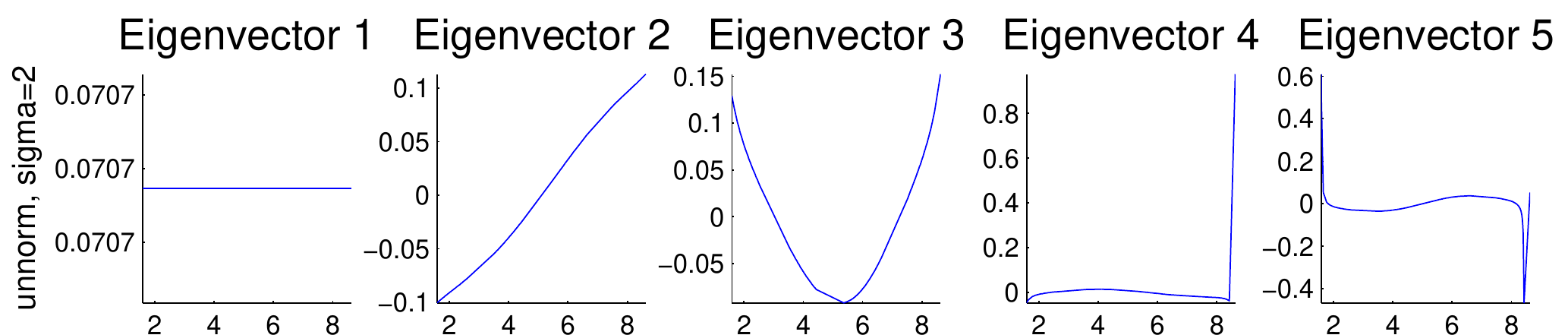}\\
\includegraphics[height=0.10\textheight]{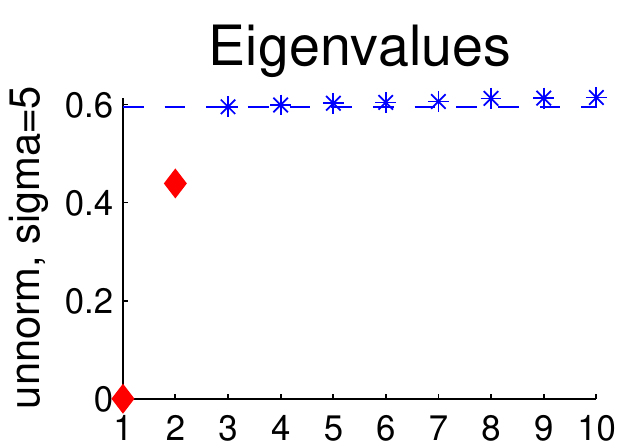}
\includegraphics[height=0.10\textheight]{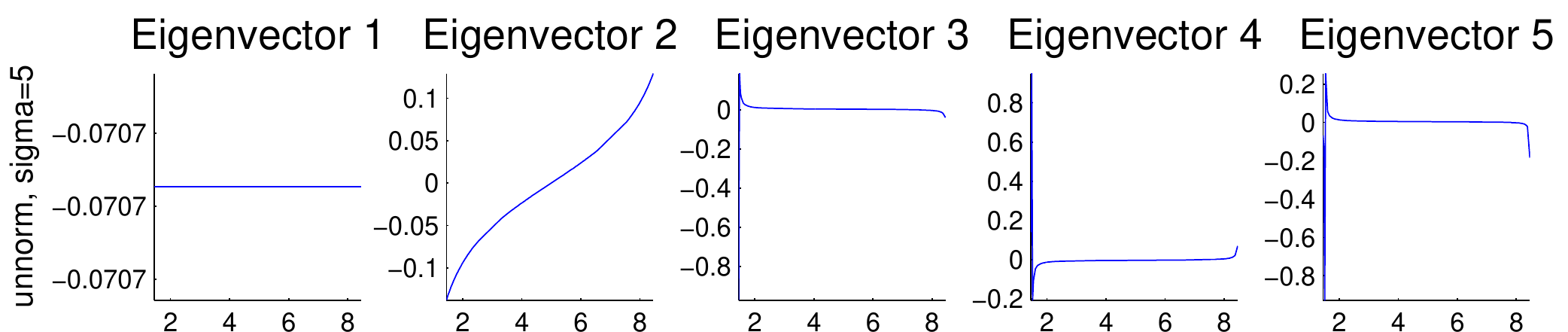}
\end{center}
\caption{Consistency of unnormalized spectral clustering. Plotted are
  eigenvalues and eigenvectors of $\Lun$, for parameter $\sigma=2$ (first
  row) and $\sigma=5$ (second row). The dashed line indicates $\min
  d_j$, the eigenvalues below $\min d_j$ are plotted as red diamonds,
  the eigenvalues above $\min d_j$ are plotted as blue stars.  See text for more details.}
\label{fig-counterex}
\end{figure}

A completely different argument for the superiority of normalized
spectral clustering comes from a statistical analysis of both
algorithms. In a statistical setting one assumes that the data
points $x_1, \hdots, x_n$ have been sampled i.i.d. according to some
probability distribution $P$ on some underlying data space
$\Xcal$. The most fundamental question is then the question of
consistency: if we draw more and more data points, do the clustering
results of spectral clustering converge to a useful partition of the
underlying space $\Xcal$?\\

For both normalized spectral clustering algorithms, it can be proved
that this is indeed the case
\cite{LuxBouBel04_colt,LuxBouBel04_nips,LuxBelBou_annofstats}. Mathematically,
one proves that as we take the limit $n \to \infty$, the matrix
$\Lsym$ converges in a strong sense to an operator $U$ on the space
$C(\Xcal)$ of continuous functions on $\Xcal$. This convergence
implies that the eigenvalues and eigenvectors of $\Lsym$ converge to
those of $U$, which in turn can be transformed to a statement about
the convergence of normalized spectral clustering. One can show that
the partition which is induced on $\Xcal$ by the eigenvectors of $U$
can be interpreted similar to the random walks interpretation of
spectral clustering. That is, if we consider a diffusion process on
the data space $\Xcal$, then the partition induced by the eigenvectors
of $U$ is such that the diffusion does not transition between the
different clusters very often \cite{LuxBouBel04_colt}. All consistency
statements about normalized spectral clustering hold, for both $\Lsym$
and $\Lrw$, under very mild conditions which are usually satisfied in
real world applications. Unfortunately, explaining more details about
those results goes beyond the scope of this tutorial, so we refer the
interested reader to \citeA{LuxBelBou_annofstats}.\\

In contrast to the clear convergence statements for normalized
spectral clustering, the situation for unnormalized spectral
clustering is much more unpleasant.  It can be proved that
unnormalized spectral clustering can fail to converge, or that it can
converge to trivial solutions which construct clusters consisting of
one single point of the data space \cite{LuxBouBel04_nips,LuxBelBou_annofstats}. 
Mathematically, even though one can prove that the
matrix $(1/n)\Lun$ itself converges to some limit operator $T$ on
$C(\Xcal)$ as $n \to \infty$, the spectral properties of this limit
operator $T$ can be so nasty that they prevent the convergence of
spectral clustering. It is possible to construct examples
which show that this is not only a problem for very large sample size,
but that it can lead to completely unreliable results even for small sample
size. At least it is possible to characterize the conditions when
those problem do not occur: We have to make sure that the eigenvalues
of $\Lun$ corresponding to the eigenvectors used in unnormalized
spectral clustering are significantly smaller than the minimal degree in the
graph. This means that if we use the first $k$ eigenvectors for
clustering, then $\lambda_i \ll \min_{j=1,\hdots,n} d_j$ should hold
for all $i=1,\hdots,k$.  The mathematical reason for this condition is
that eigenvectors corresponding to eigenvalues larger than $\min d_j$
approximate Dirac functions, that is they are approximately 0 in all
but one coordinate. If those eigenvectors are used for clustering,
then they separate the one vertex where the eigenvector is non-zero
from all other vertices, and we clearly do not want to construct such
a partition. Again we refer to the literature for precise statements
and proofs.\\

For an illustration of this phenomenon, consider again our toy data set
from Section \ref{sec-algorithms}.  We consider the first eigenvalues
and eigenvectors of the unnormalized graph Laplacian based on the
fully connected graph, for different choices of the parameter $\sigma$
of the Gaussian similarity function (see last row of Figure
\ref{fig-example-spectral} and all rows of Figure \ref{fig-counterex}).  The
eigenvalues above $\min d_j$ are plotted as blue stars, the
eigenvalues below $\min d_j$ are plotted as red diamonds.  The dashed
line indicates $\min d_j$.
In general, we can see that the eigenvectors corresponding to
eigenvalues which are much below the dashed lines are ``useful''
eigenvectors.  In case $\sigma=1$ (plotted already in the last row of
Figure \ref{fig-example-spectral}), Eigenvalues 2, 3 and 4 are
significantly below $\min d_j$, and the corresponding Eigenvectors 2,
3, and 4 are meaningful (as already discussed in Section
\ref{sec-algorithms}).  If we increase the parameter $\sigma$, we can
observe that the eigenvalues tend to move towards $\min d_j$. In case
$\sigma=2$, only the first three eigenvalues are below $\min d_j$
(first row in Figure \ref{fig-counterex}), and in case $\sigma=5$ only
the first two eigenvalues are below $\min d_j$ (second row in Figure
\ref{fig-counterex}). We can see that as soon as an eigenvalue gets
close to or above $\min d_j$, its corresponding eigenvector
approximates a Dirac function. Of course, those eigenvectors are
unsuitable for constructing a clustering. In the limit for $n \to
\infty$, those eigenvectors would converge to perfect Dirac
functions. Our illustration of the finite sample case shows that this
behavior not only occurs for large sample size, but can be generated
even on the small example in our toy data set.\\

It is very important to stress that those problems only concern the
eigenvectors of the matrix $\Lun$, and they do not occur for $\Lrw$ or
$\Lsym$. Thus, from a statistical point of view, it is preferable to
avoid unnormalized spectral clustering and to use the normalized
algorithms instead.\\

\subsubsection*{Which normalized Laplacian?}

Looking at the differences between the two normalized spectral
clustering algorithms using $\Lrw$ and $\Lsym$, all three explanations
of spectral clustering are in favor of $\Lrw$.  The reason is that the
eigenvectors of $\Lrw$ are cluster indicator vectors $\charfct_{A_i}$,
while the eigenvectors of $\Lsym$ are additionally multiplied with
$D^{1/2 }$, which might lead to undesired artifacts. %
As using $\Lsym$ also
does not have any computational advantages, we thus advocate for using
$\Lrw$. \\

\section{Outlook and further reading} \label{sec-extensions}

Spectral clustering goes back to \citeA{DonHof73}, who first suggested
to construct graph partitions based on eigenvectors of the adjacency
matrix. In the same year, \citeA{Fiedler73} discovered that
bi-partitions of a graph are closely connected with the second
eigenvector of the graph Laplacian, and he suggested to use this
eigenvector to partition a graph. Since then, spectral clustering has
been discovered, re-discovered, and extended many times in different
communities, see for example 
\citeA{PotSimLio90},
\citeA{Simon91},
\citeA{Bolla91},
\citeA{HagKah92},  %
\citeA{HenLel95},  %
\citeA{DriRoo95},  %
\citeA{BarPotSim95},
\citeA{SpiTen96},
\citeA{GuaMil98}. 
A nice overview over the history of spectral clustering can be found
in \citeA{SpiTen96}. \\

In the machine learning community, spectral clustering has been made popular
by the works of \citeA{ShiMal00},
\citeA{NgJorWei02}, 
\citeA{MeiShi01}, and \citeA{Ding04_tutorial}. Subsequently, spectral clustering has been extended to many non-standard settings, for example 
spectral clustering applied to the co-clustering problem
\cite{Dhillon01}, 
spectral clustering with additional side information \cite{Joachims03}
connections between spectral clustering and the weighted kernel-$k$-means
algorithm \cite{DhiGuaKul05}, 
learning similarity functions based on spectral clustering 
\cite{BachJordan04}, 
or spectral clustering in a distributed environment
\cite{KemMcs04}. 
Also, new theoretical insights about the relation of spectral
clustering to other algorithms have been found. A link between
spectral clustering and the weighted kernel $k$-means algorithm is
described in \citeA{DhiGuaKul05}. Relations between spectral clustering
and (kernel) principal component analysis rely on the fact that the
smallest eigenvectors of graph Laplacians can also be interpreted as
the largest eigenvectors of kernel matrices (Gram matrices). Two
different flavors of this interpretation exist: while
\citeA{BenDelRouEtal04}
 interpret the matrix $D^{-1/2}W D^{-1/2}$ as
kernel matrix, other authors \cite{SaeFouYenDup04} interpret the
Moore-Penrose inverses of $\Lun$ or $\Lsym$ as kernel matrix. Both
interpretations can be used to construct (different) out-of-sample
extensions for spectral clustering. Concerning application cases of
spectral clustering, in the last few years such a huge number of
papers has been published in various scientific areas that it is
impossible to cite all of them. We encourage the reader to query his
favorite literature data base with the phrase ``spectral clustering''
to get an impression no the variety of applications.\\

The success of spectral clustering is mainly based on the fact that it
does not make strong assumptions on the form of the clusters. As
opposed to $k$-means, where the resulting clusters form convex
sets (or, to be precise, lie in disjoint convex sets of the underlying space), spectral
clustering can solve very general problems like intertwined
spirals. Moreover, spectral clustering can be implemented efficiently
even for large data sets, as long as we make sure that the similarity
graph is sparse. Once the similarity graph is chosen, we just have to
solve a linear problem, and there are no issues of getting stuck in
local minima or restarting the algorithm for several times with
different initializations. However, we have already mentioned that
choosing a good similarity graph is not trivial, and spectral
clustering can be quite unstable under different choices of the
parameters for the neighborhood graphs.
So spectral clustering cannot serve as a ``black box algorithm'' which
automatically detects the correct clusters in any given data set. But
it can be considered as a powerful tool which can produce 
good results if applied with care.\\

In the field of machine learning, graph Laplacians are not only used
for clustering, but also emerge for many other tasks such as
semi-supervised learning (e.g., \citeNP{ChaZieSch06} for an overview)
or manifold reconstruction (e.g., \citeNP{BelNiy03eigenmaps}). In most
applications, graph Laplacians are used to encode the assumption that
data points which are ``close'' (i.e., $w_{ij}$ is large) should have
a ``similar'' label (i.e., $f_i \approx f_j$). A function $f$
satisfies this assumption if $w_{ij}(f_i - f_j)^2$ is small for all
$i,j$, that is $f'\Lun f$ is small. With this intuition one can use
the quadratic form $f'\Lun f$ as a regularizer in a transductive
classification problem. One other way to interpret the use of graph
Laplacians is by the smoothness assumptions they encode.  A function
$f$ which has a low value of $f'\Lun f$ has the property that it
varies only ``a little bit'' in regions where the data points lie
dense (i.e., the graph is tightly connected), whereas it is allowed to
vary more (e.g., to change the sign) in regions of low data density.
In this sense, a small value of $f'\Lun f$ encodes the so called
``cluster assumption'' in semi-supervised learning, which requests
that the decision boundary of a
classifier should lie in a region of low density. \\

An intuition often used is that graph Laplacians formally look like a
continuous Laplace operator (and this is also where the name ``graph
Laplacian'' comes from). To see this, transform a local similarity
$w_{ij}$ to a distance $d_{ij}$ by the relationship $w_{ij} = 1 /
d_{ij}^2$ and observe that \ba w_{ij}(f_i - f_j)^2 \approx \left(
  \frac{f_i - f_j}{d_{ij}} \right)^2 \ea
looks like a difference quotient. As a consequence, the equation $f'Lf
= \sum_{ij} w_{ij} (f_i - f_j)^2$ from Proposition \ref{prop-unnorm}
looks like a discrete version of the quadratic form associated to the
standard Laplace operator $\Lcal$ on $\R^n$, which satisfies \ba
\langle g, \Lcal g \rangle = \int | \nabla g |^2 dx.  \ea
This intuition has been made precise in the works of \citeA{Belkin03}, \citeA{Lafon04},
\citeA{HeiAudLux05,HeiAudLux07}, \citeA{BelNiy05}, \citeA{Hein06},
\citeA{GinKol06}.
In general, it is proved that graph Laplacians are
discrete versions of certain continuous Laplace operators, and that if
the graph Laplacian is constructed on a similarity graph of randomly
sampled data points, then it converges to some continuous Laplace
operator (or Laplace-Beltrami operator) on the underlying space.
\citeA{Belkin03} studied the first important step of the convergence
proof, which deals with the convergence of a continuous operator
related to discrete graph Laplacians to the Laplace-Beltrami
operator. %
His results were generalized from uniform distributions to general
distributions by \citeA{Lafon04}. Then in \citeA{BelNiy05}, the
authors prove pointwise convergence results for the unnormalized graph
Laplacian using the Gaussian similarity function on manifolds with
uniform distribution. At the same time, \citeA{HeiAudLux05} prove more
general results, taking into account all different graph Laplacians
$L$, $\Lrw$, and $\Lsym$, more general similarity functions, and
manifolds with arbitrary distributions. In \citeA{GinKol06},
distributional and uniform convergence results are proved on manifolds
with uniform distribution.  \citeA{Hein06} studies the convergence of
the smoothness functional induced by the graph
Laplacians and shows uniform convergence results. \\

Apart from applications of graph Laplacians to partitioning problems
in the widest sense, graph Laplacians can also be used for completely
different purposes, for example for graph drawing \cite{Koren05}. In
fact, there are many more tight connections between the topology and
properties of graphs and the graph Laplacian matrices than we have
mentioned in this tutorial. Now equipped with an understanding for the
most basic properties, the interested reader is invited to further
explore and enjoy the huge literature in this field on his own.

\bibliography{general_bib,ules_publications}
\end{document}